\pdfoutput=1
\documentclass[12pt]{article}

\RequirePackage{amsthm,amsmath}
\RequirePackage[authoryear]{natbib}

\usepackage{amssymb}
\usepackage{amsfonts}

\usepackage{colortbl}
\usepackage{color}
\usepackage{xcolor}
\usepackage{array}
\usepackage{multirow, multicol}

\usepackage{graphicx}
\usepackage{picinpar}

\RequirePackage[colorlinks,citecolor=blue,urlcolor=blue]{hyperref}

\newcommand\labelA[1]{\label{#1}}

\newcommand\comment[1]{}

\newtheorem{conj}{Conjecture}
\newtheorem{thm}[conj]{Theorem}
\newtheorem{cor}[conj]{Corollary}
\newtheorem{prop}[conj]{Proposition}

\newtheorem{definition}[conj]{Definition}

\newtheorem{thmref}{Theorem}

\def\paper{paper }

\newcommand{\topr}{\,\mathop{\longrightarrow}\limits^{P}\,}
\newcommand{\eqa}{\,\mathop{=}\limits^{l}\,}
\newcommand{\leqa}{\,\mathop{\leq}\limits^{l}\,}
\newcommand{\ga}{\,\mathop{>}\limits^{l}\,}

\def\reals{{\mathbb R}}

\begin{document}


\title{Reconstruction of a Low-rank Matrix in the Presence of Gaussian Noise}
\author{Andrey Shabalin and Andrew Nobel}

\maketitle

\comment{

\author{\fnms{Andrey} \snm{Shabalin}\ead[label=e1]{shabalin@email.unc.edu}}
\address{\printead{e1}}
\and
\author{\fnms{Andrew} \snm{Nobel}\ead[label=e2]{nobel@email.unc.edu}}
\address{\printead{e2}}
\affiliation{University of North Carolina at Chapel Hill}

\runauthor{A. Shabalin and A. Nobel}
}

\begin{abstract}
In this paper we study the problem of reconstruction of a low-rank matrix observed with additive Gaussian noise.
First we show that under mild assumptions (about the prior distribution of the signal matrix) we can restrict our attention to reconstruction methods that are based on the singular value decomposition of the observed matrix and act only on its singular values (preserving the singular vectors).
Then we determine the effect of noise on the SVD of low-rank matrices by building a connection between matrix reconstruction problem and spiked population model in random matrix theory. 
Based on this knowledge, we propose a new reconstruction method, called RMT, that is designed to reverse the effect of the noise on the singular values of the signal matrix and adjust for its effect on the singular vectors. 
With an extensive simulation study we show that the proposed method outperform even oracle versions of both soft and hard thresholding methods and closely matches the performance of a general oracle scheme.
\end{abstract}





\section{Introduction%
\labelA{sec:intro}}

Existing and emerging technologies provide scientists with access to a growing wealth of data.
Some data is initially produced in the matrix form, while other can be represented in a matrix form once the data from multiple samples is combined.
The data is often measured with noise due to limitations of the data generating technologies.
To reduce the noise we need some information about the possible structure of signal component. In this paper we assume the signal to be a low rank matrix.
Assumption of this sort appears in multiple fields of study including genomics \citep{
raychaudhuri2000principal,
alter2000singular,
holter2000fundamental,
wall2001svdman,
troyanskaya2001missing}, compressed sensing \citep{candes2006robust,
candes2008exact,
donoho2006compressed}, and image denoising (\citeauthor{wongsawat0000multichannel}; \citealp{konstantinides1997noise}). 
In many cases the signal matrix is known to have low rank. 
For example, a matrix of squared distances between points in $d$-dimensional Euclidean space is know to have rank at most $d+2$.
A correlation matrix for a set of points in $d$-dimensional Euclidean space has rank at most $d$.
In other cases the target matrix is often assumed to have low rank, or to have a good low-rank approximation. 

In this paper we address the problem of recovering a low rank signal matrix whose entries are observed in the presence of additive Gaussian noise.
The reconstruction problem considered here has a signal plus noise structure.
Our goal is to recover an unknown $m \times n$ matrix $A$ of low rank that is observed in the presence of i.i.d. Gaussian noise as matrix $Y$:
\begin{equation*}
\labelA{eq:modelvar}
Y \, = \, A + \frac{\sigma}{\sqrt{n}} W, 
\qquad \mbox{ where } \ W_{ij} \, \sim \, \mbox{ i.i.d.\ } N(0,1).
\end{equation*}
The factor $n^{-1/2}$ ensures that the signal and noise are comparable, and is
employed for the asymptotic study of matrix reconstruction in Section \ref{sec:asy}.
In what follows, we first consider the variance of the noise $\sigma^2$ to be known, and assume that it is equal to one.
(In Section \ref{sec:noisevar} we propose an estimator for $\sigma$, which we use in the proposed reconstruction method.)
In this case the model (\ref{eq:modelvar}) simplifies to
\begin{equation}
\labelA{eq:model}
Y \, = \, A + \frac{1}{\sqrt{n}} W,
\qquad \mbox{ where } \ W_{ij} \, \sim \, \mbox{ i.i.d.\ } N(0,1).
\end{equation}

Formally, a matrix recovery scheme is a map 
$g$\,$: \reals^{m \times n} \to \reals^{m \times n}$ from the 
space of $m \times n$ matrices to itself.  Given a recovery scheme $g(\cdot)$ and
an observed matrix $Y$ from the model (\ref{eq:model}), we regard 
$\widehat A = g(Y)$ as an estimate of $A$,
and measure the performance
of the estimate $\widehat A$ by 
\begin{equation}
\labelA{eq:loss}
\mbox{Loss} (A,\widehat A)
\, = \,
\| \widehat A - A \|_F^2,
\end{equation}
where $\|\cdot\|_F$ denotes the Frobenius norm. The Frobenius norm of an $m \times n$ matrix 
$B = \{ b_{ij} \}$ is given by
\begin{equation*}
\| B \|_F^2 \, = \, \sum_{i=1}^m \sum_{j=1}^n b_{ij}^2 .
\end{equation*}
Note that if the vector space $\reals^{m \times n}$ is equipped with the inner
product $\langle A, B \rangle = \mbox{tr}(A' B)$, then 
$\| B \|_F^2 = \langle B, B \rangle$.

\subsection{Hard and Soft Thresholding}

A natural starting point for reconstruction of the target matrix $A$ in (\ref{eq:model}) is the singular value decomposition (SVD) of the observed matrix $Y$.  
Recall that the singular value decomposition of an $m \times n$ matrix $Y$ is given by the factorization
\[
Y
\, = \,
U D V'
\, = \,
\sum_{j=1}^{m \wedge n} d_j u_j v_j'
.
\]
Here 
$U$ is an $m \times m$ orthogonal matrix whose columns are the left singular vectors $u_j$,
$V$ is an $n \times n$ orthogonal matrix whose columns are the right singular vectors $v_j$, and 
$D$ is an $m \times n$ matrix with singular values $d_j = D_{jj} \geq 0$ on the diagonal and all other entries equal to zero.  
Although it is not necessarily square, we 
will refer to $D$ as a diagonal matrix and write 
$D = \mbox{diag}(d_1,\ldots,d_{m \wedge n})$, where $m \wedge n$ denotes the minimum of $m$ and $n$.

\comment{
Each of the methods cited above () rely on the SVD. 
The singular value decomposition is a popular tool for the analysis of data.
For instance, \cite{wall2001svdman},  \cite{alter2000singular}, and \cite{holter2000fundamental} used SVD as a tool for mining gene expression data and \cite{troyanskaya2001missing} applied SVD to impute missing values. 
SVD is also used for image denoising. However a better performance is achieved when SVD is applied to small blocks of pixels, not to the whole image. Denoising methods of \citeauthor{wongsawat0000multichannel} and \cite{konstantinides1997noise} perform SVD on square subblocks and set to zero the singular values smaller than some threshold.
}

Many matrix reconstruction schemes 
act by shrinking the singular values of the observed matrix towards zero.
Shrinkage is typically accomplished by hard or soft thresholding.
Hard thresholding schemes set every singular value of $Y$ less than 
a given positive threshold $\lambda$ equal to zero, leaving other singular values unchanged.  
The family of hard thresholding schemes is defined by
\[
g_\lambda^H(Y) \, = \, \sum_{j=1}^{m \wedge n} d_j I( d_j \geq \lambda) \, u_j v_j',
\qquad \lambda > 0 .
\]
Soft thresholding schemes subtract a given positive number $\nu$ from each singular value, setting values less than $\nu$ equal to zero.
The family of soft thresholding schemes is defined by
\[
g_\nu^S(Y) \, = \, \sum_{j=1}^{m \wedge n} ( d_j - \nu)_+ \, u_j v_j' ,
\qquad \nu > 0 .
\]
Hard and soft thresholding schemes can be defined equivalently in the penalized forms
\[
g_\lambda^H(Y)
\, = \,
\mathop{\arg\min}_{B} \big\{ \| Y - B \|_F^2 + \lambda^2 \, \mbox{rank}(B) \big\}
\]
\[
g_\nu^S(Y)
\ = \ 
\mathop{\arg\min}_{B}  \big\{ \| Y - B \|_F^2 + 2 \nu \, \|B\|_* \big\} .
\]
In the second display, $\| B \|_{*}$ denotes the nuclear norm 
of $B$, which is equal to the sum of its singular values.

\begin{figure}[htbp]
\begin{center}
\includegraphics[width = 6cm]{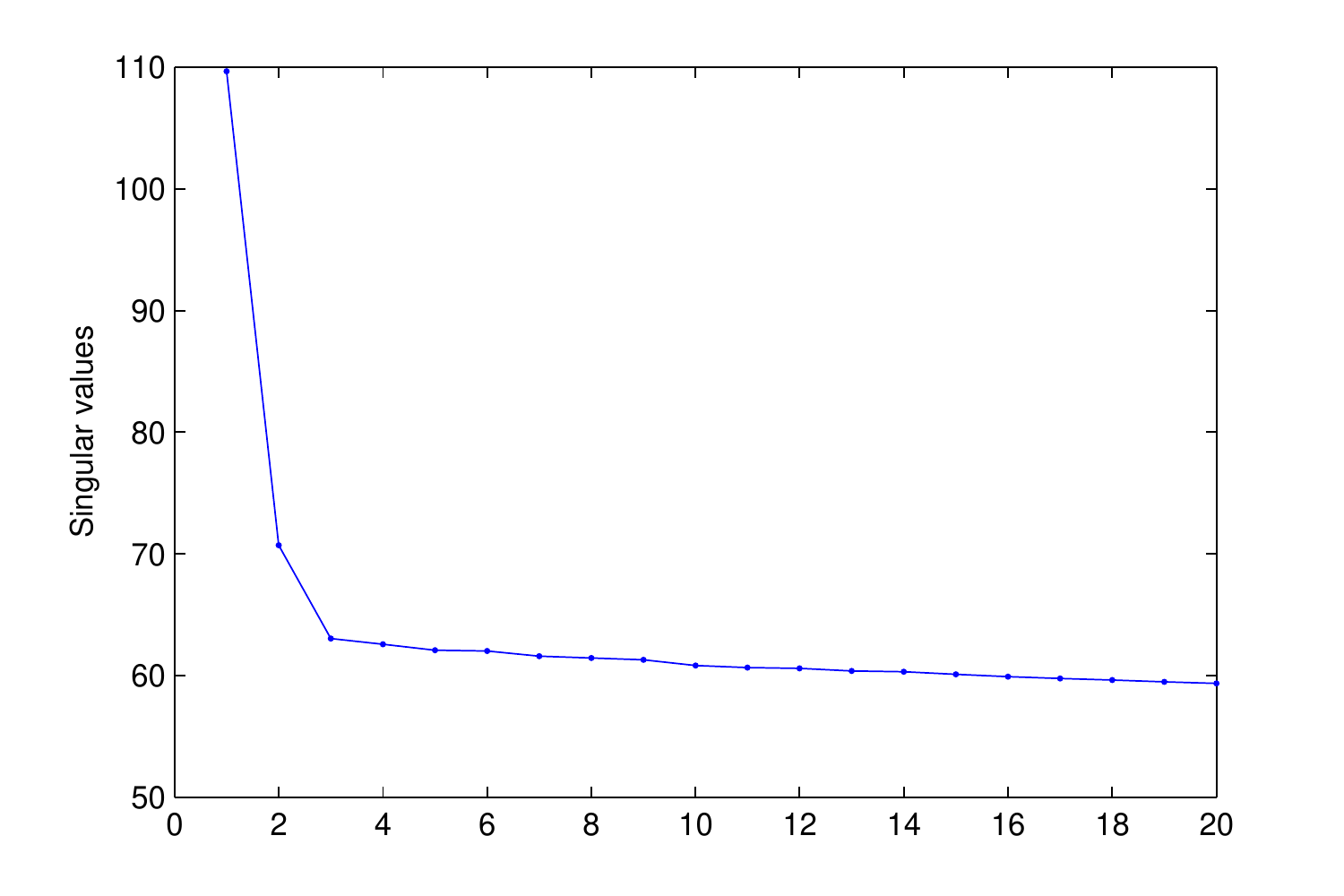}%
\caption[Scree plot example.]%
{Scree plot for a $1000 \times 1000$ rank 2 signal matrix with noise.}
\label{fig:Scree}
\end{center}
\end{figure}

In practice, hard and soft thresholding schemes require estimates of the
noise variance, as well as the selection of appropriate cutoff or 
shrinkage parameters.  
There are numerous methods in the literature for choosing the 
hard threshold $\lambda$.
Heuristic methods often make use of 
the scree plot, which displays the singular values of $Y$ 
in decreasing order: $\lambda$ is typically chosen to 
be the y-coordinate of a well defined ``elbow'' in the resulting curve.
A typical scree plot for a $1000 \times 1000$ matrix with rank 2 signal is shown in Figure~\ref{fig:Scree}.
The ``elbow'' point of the curve on the plot clearly indicate that the signal has rank 2.

A theoretically justified selection of hard threshold $\lambda$ is presented in recent work of \cite{bunea2010adaptive}. They also provide performance guaranties for the resulting hard thresholding scheme using techniques from empirical process theory and complexity regularization.  
Selection of the soft thresholding shrinkage parameter $\nu$
may also be accomplished by a variety of methods.  
\cite{negahban2009estimation} propose
a specific choice of $\nu$ and provide performance guarantees 
for the resulting soft thresholding scheme.

\begin{figure}[htbp]
\begin{center}
\includegraphics[width=6cm]{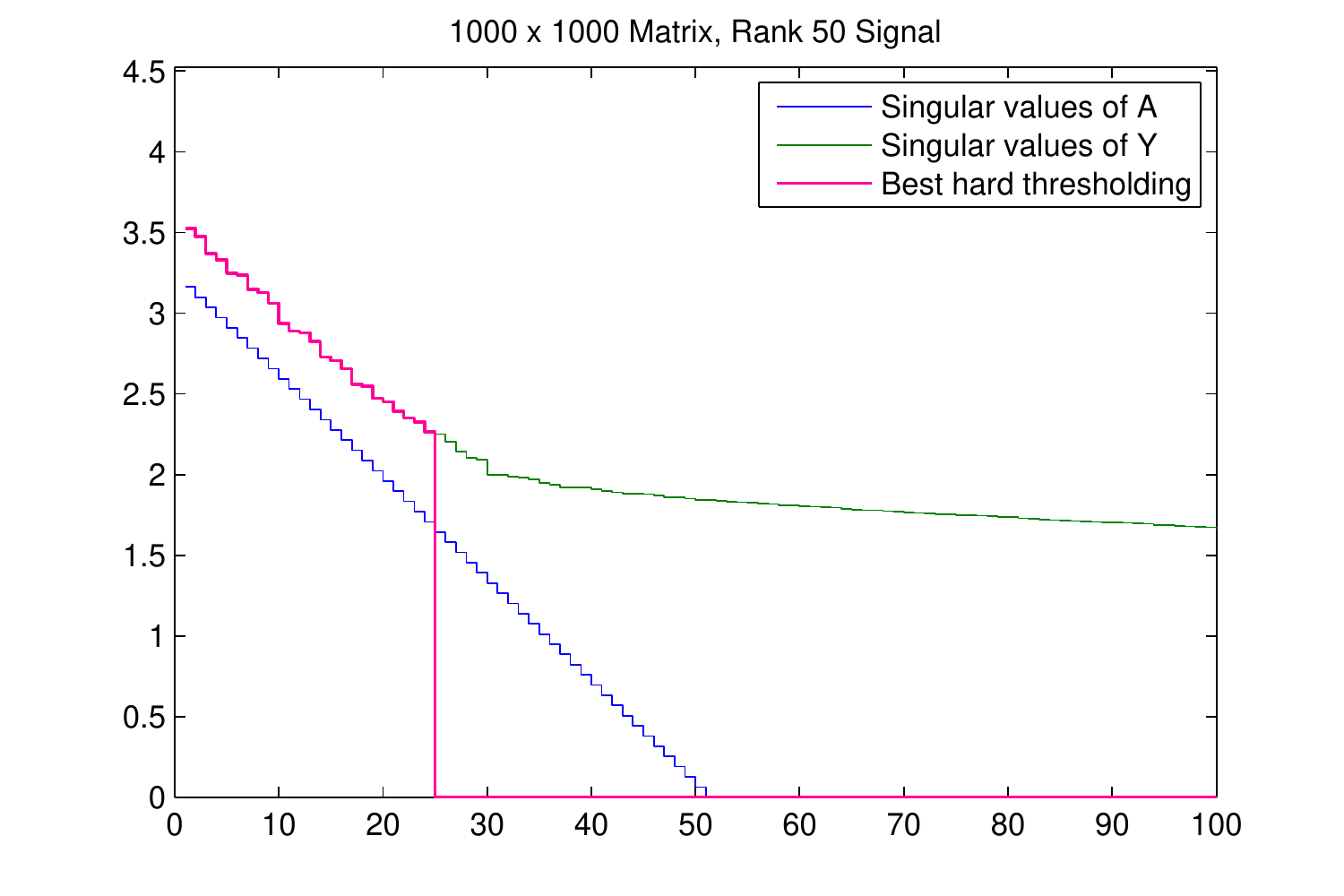}
\quad
\includegraphics[width=6cm]{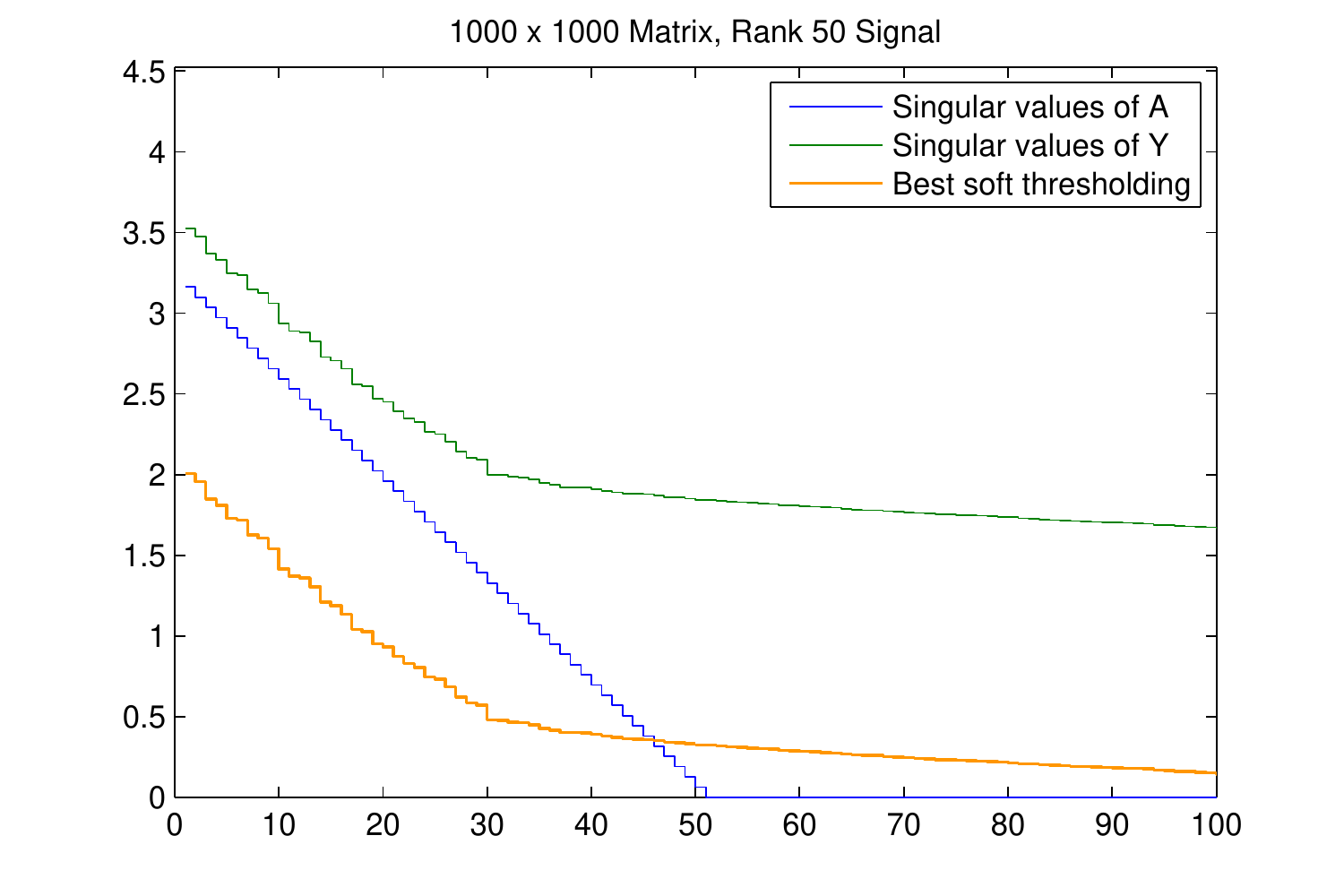}
\caption{Singular values of hard and soft thresholding estimates.}
\label{fig:sinvals}
\end{center}
\end{figure}

Figure \ref{fig:sinvals} illustrates the action of hard and soft thresholding on a $1000 \times 1000$ matrix 
with a rank 50 signal. The blue line indicates the singular values of the signal $A$ and the green 
line indicates the those of the observed matrix $Y$. 
The plots show the singular values of the hard and soft thresholding estimates 
incorporating the best choice of the parameters $\lambda$ and $\nu$, respectively. 
It is evident from the figure that neither thresholding scheme delivers an accurate estimate of the signal's singular values. Moreover examination of the loss indicates that they do not provide a good estimates of the signal matrix. 

\vskip 0.1in

The families of hard and soft thresholding methods encompass many existing reconstruction schemes.
Both thresholding approaches seek low rank (sparse) estimates of the target matrix, 
and both can be naturally formulated as optimization problems.
However, the family of all reconstruction schemes is much larger, and it is natural to consider
alternatives to hard and soft thresholding that may offer better performance. 

In this paper, we start with a principled analysis of the matrix reconstruction problem,
with the effort of making as few assumptions as possible.
Theoretically motivated design of the method.
Based on the analysis of the reconstruction problem and,
in particular, the analysis of the effect of noise on low-rank matrices we
design a new reconstruction method with a theoretically motivated design.

\subsection{Outline} 

We start the paper with an analysis of the finite sample properties of the matrix reconstruction problem.
The analysis does not require the matrix $A$ to have low rank and only requires that the distribution of noise matrix $W$ is orthogonally invariant (does not change under left and right multiplications by orthogonal matrices).
Under mild conditions (the prior distribution of $A$ must be orthogonally invariant)
we prove that we can restrict our attention to the reconstruction methods that are based on the SVD of the observed matrix $Y$ and act only on its singular values, not affecting the singular vectors.
This result has several useful consequences.
First, it reduces the space of reconstruction schemes we consider from $g:\reals^{m \times n} \to \reals^{m \times n}$ to just $\reals^{m \wedge n} \to \reals^{m \wedge n}$.
Moreover, it gives us the recipe for design of the new reconstruction scheme: we determine the effect of the noise on the singular values and the singular value decomposition of the signal matrix $A$ and then built the new reconstruction method to reverse the effect of the noise on the singular values of $A$ and account for its effect on the singular vectors.

To determine the effect of noise on low-rank signal we build a connection between the matrix reconstruction problem and spiked population models in random matrix theory. Spiked population models were introduced by \cite{johnstone2001distribution}.
The asymptotic matrix reconstruction model that matches the setup of spiked population models 
assumes the rank of $A$ and its non-zero singular values to be fixed as the matrix dimensions to grow at the same rate: $m,n \to \infty$ and $m/n \to c>0$.
We use results from random matrix theory about the limiting behavior of the eigenvalues \citep{marcenko1967distribution,wachter1978strong,geman1980limit,baik2006els} and eigenvectors 
\citep{paul2007asymptotics,nadler2008finite,lee2010convergence}
of sample covariance matrices in spiked population models to determine the limiting behavior of the singular values and the singular vectors of matrix $Y$.

We apply these results to design a new matrix reconstruction method, which we call RMT for its use of random matrix theory. The method estimated the singular values of $A$ from the singular values of $Y$ and applies additional shrinkage to them to correct for the difference between the singular vectors of $Y$ and those of $A$. The method uses an estimator of the noise variance which is based on the sample distribution of the singular values of $Y$ corresponding to 
the zero singular values of $A$. 

We conduct an extensive simulation study to compare RMT method against the oracle versions of hard and soft thresholding and the orthogonally invariant oracle.  We run all four methods on matrices of different sizes, with signals of various ranks and spectra. The simulations clearly show that RMT method strongly outperforms the oracle versions of both hard and soft thresholding methods and closely matches the performance of the orthogonally invariant oracle (the oracle scheme that acts only on the singular values of the observed matrix).

\vskip 0.1in

The paper is organized as follows. 
In Section \ref{sec:invariant} we present the analysis of the finite sample properties of the reconstruction problem.
In Section \ref{sec:asy} we determine the effect of the noise on the singular value decomposition of row rank matrices. 
In Section \ref{sec:RMTmethod} we construct the proposed reconstruction method based on the results of the previous section. The method employs the noise variance estimator presented in Section \ref{sec:noisevar}. Finally, in Section \ref{sec:simulations} we present the simulation study comparing RMT method to the oracle versions of hard and soft thresholding and the orthogonally invariant oracle method.

\section{Orthogonally Invariant Reconstruction\\ Methods%
\labelA{sec:invariant}}
The additive model (\ref{eq:model}) and Frobenius loss (\ref{eq:loss}) 
have several elementary invariance properties, that
lead naturally to the consideration of reconstruction methods with analogous forms of invariance.
Recall that a square matrix $U$ is said to be orthogonal if $UU' = U'U = I$,
or equivalently, if the rows (or columns) of $U$ are orthonormal.  
If we multiply each side of (\ref{eq:model}) from the left  
right by orthogonal matrices $U$ and $V'$ of appropriate dimensions,
we obtain
\begin{equation}
\labelA{eq:rotated}
UYV' \, = \, UAV' \,+\, \frac{1}{\sqrt{n}} UWV' .
\end{equation}

\begin{prop} 
\labelA{thm:orthogonal}
Equation (\ref{eq:rotated}) is a reconstruction problem of the form (\ref{eq:model})
with signal $UAV'$ and observed matrix $UYV'$.  
If $\widehat A$ is an estimate of $A$ in model (\ref{eq:model}), 
then $U \widehat A V'$ is an estimate of $UAV'$ in model (\ref{eq:rotated}) 
with the same loss.
\end{prop}

\begin{proof}
If $A$ has rank $r$ then $UAV'$ also has rank $r$.
Thus prove the first statement, it suffices to show that 
$UWV'$ in (\ref{eq:rotated}) has independent $N(0,1)$ entries. This follows from standard properties of the multivariate normal distribution.
In order to establish the second statement of the proposition, let $U$ and 
$V$ be the orthogonal matrices in (\ref{eq:rotated}).
For any $m \times n$ matrix $B$,
\begin{equation*}
\| U B \|_F^2 
\, = \,
\mbox{tr} \big[ (U B)' (U B) \big]
\, = \,
\mbox{tr} \big[ B' B \big]
\, = \,
\| B \|_F^2 ,
\end{equation*}
and more generally $\|U B V' \|_F^2 = \| B \|_F^2$. 
Applying the last equality to $B = \widehat A - A$ yields
\begin{equation*}
\mbox{Loss} (UAV', U \widehat AV') 
\, = \,
\| U (\widehat A - A)V' \|_F^2 
\, = \,
\| \widehat A - A \|_F^2 
\, = \,
\mbox{Loss} (A, \widehat A)
\end{equation*}
as desired.
\end{proof}

In the proof we use the fact that the distribution of matrix $W$ does not change
under left and right multiplications by orthogonal matrices. We will call such distributions orthogonally invariant.

\begin{definition} 
\labelA{def:ortinvdist}
A random $m \times n$ matrix $Z$ has 
an orthogonally invariant distribution
if for any orthogonal matrices $U$ and $V$ of appropriate size the distribution of $UZV'$ 
is the same as the distribution of $Z$.
\end{definition}


In light of Proposition \ref{thm:orthogonal} 
it is natural to consider reconstruction schemes those action
do not change under orthogonal transformations of the reconstruction problem.

\begin{definition}
\labelA{def:ortinvscheme}
A reconstruction scheme $g(\cdot)$ is orthogonally invariant 
if for any $m \times n$ matrix $Y$, and any orthogonal matrices 
$U$ and $V$ of appropriate size,
$g(UYV') 
=
U g(Y) V'.$
\end{definition}

In general, a good reconstruction method need not be orthogonally
invariant.  For example, if the signal matrix $A$ is known
to be diagonal, then for each $Y$ the estimate $g(Y)$ should be
diagonal as well, and in this case $g(\cdot)$ is not orthogonally
invariant.
However, as we show in the next theorem, if we have no information 
about the singular vectors of $A$ (either prior information or information from the singular values of $A$), then it suffices to restrict our attention to orthogonally invariant reconstruction schemes.  

\vskip.2in

\begin{thm} 
\labelA{thm:ortoptimality} 
Let $Y = \mathbf{A} + W$, where $\mathbf{A}$ is a random target matrix. 
Assume that $\mathbf{A}$ and $W$ are independent and have 
orthogonally invariant distributions.
Then, for every reconstruction scheme $g(\cdot)$,
there is an orthogonally invariant reconstruction 
scheme $\tilde g(\cdot)$ whose expected loss is the same, or smaller, 
than that of $g(\cdot)$.
\end{thm}

\begin{proof}
Let $\mathbf{U}$ be an $m \times m$ random matrix that is independent of 
$\mathbf{A}$ and $W$, and is distributed according to 
Haar measure on the compact group of $m \times m$ orthogonal
matrices.  Haar measure is (uniquely) defined by the
requirement that, for every $m \times m$ orthogonal matrix $C$, 
both $C \mathbf{U}$ and $\mathbf{U} C$
have the same distribution as $\mathbf{U}$ 
\citep[\textit{c.f.}][]{hofmann2006structure}. 
Let $\mathbf{V}$ 
be an $n \times n$ random matrix distributed according to the 
Haar measure on the compact group of $n \times n$ orthogonal
matrices that is independent of 
$\mathbf{A}$, $W$ and $\mathbf{U}$.
Given a reconstruction scheme $g(\cdot)$, define a new  scheme 
\[
\tilde g(Y) 
\, = \, 
\mathbb{E} [\mathbf{U}' g(\mathbf{U}Y\mathbf{V}') \mathbf{V} \, | \, Y] .
\]
It follows from the definition of $\mathbf{U}$ and $\mathbf{V}$ that
$\tilde g(\cdot)$ is orthogonally invariant.  The independence
of $\{ \mathbf{U}, \mathbf{V} \}$ and $\{ \mathbf{A}, W \}$ ensures 
that conditioning on $Y$ is equivalent to conditioning on $\{\mathbf{A},W\}$,
which yields the equivalent representation
\[
\tilde g(Y) 
\, = \,
\mathbb{E} [\mathbf{U}' g(\mathbf{U}Y\mathbf{V}') \mathbf{V} 
\, | \, \mathbf{A},W]
.
\] 
Therefore,
\begin{equation*}
\begin{split}
\mathbb{E} \, \mbox{Loss}(\mathbf{A},\tilde g(Y)) 
\, & = \,
\mathbb{E} \, \big\| 
\mathbb{E} [ \mathbf{U}' g(\mathbf{U} Y \mathbf{V}') \mathbf{V} - \mathbf{A} 
\, | \, \mathbf{A}, W ]  \big\|_F^2 \\ 
& \leq \,
\mathbb{E} \|  \mathbf{U}' g(\mathbf{U} Y \mathbf{V}')\mathbf{V} - \mathbf{A}  \|_F^2 \\
& = \,
\mathbb{E} \|  g(\mathbf{U} Y \mathbf{V}') - \mathbf{U}\mathbf{A}\mathbf{V}'  \|_F^2
,
\end{split}
\end{equation*}
the inequality follows from the conditional version of 
Jensen's inequality applied to each term in the sum defining the squared norm.
The final equality follows from the orthogonality of 
$\mathbf{U}$ and $\mathbf{V}$.  The last term in the previous display can be
analyzed as follows:
\begin{equation*}
\begin{split}
\mathbb{E} 
\big\|  g(\mathbf{U} Y \mathbf{V}') -  \mathbf{U}&\mathbf{A}\mathbf{V}'  \big\|_F^2 
\\
\, & = \,
\mathbb{E} \Big[ \mathbb{E} 
\big(\|  g(\mathbf{U} \mathbf{A} \mathbf{V}'+n^{-1/2}\mathbf{U} W \mathbf{V}') - 
\mathbf{U}\mathbf{A}\mathbf{V}'  \|_F^2 
\, | \, \mathbf{U}, \mathbf{V},\mathbf{A} \big) \Big] \\ 
& = \,
\mathbb{E} \Big[ \mathbb{E} 
\big(\| g(\mathbf{U} \mathbf{A} \mathbf{V}'+n^{-1/2}W) - 
\mathbf{U} \mathbf{A} \mathbf{V}'  \|_F^2 \, | \,  
\mathbf{U}, \mathbf{V},\mathbf{A} \big) \Big] \\
& = \,
\mathbb{E} \big\|  g(\mathbf{U} \mathbf{A} \mathbf{V}'+n^{-1/2}W) - 
\mathbf{U}\mathbf{A}\mathbf{V}'  \big\|_F^2
.
\end{split}
\end{equation*}
The first equality follows from the definition of $Y$;
the second follows from the independence of
$W$ and $\mathbf{U}, \mathbf{A}, \mathbf{V}$, and the orthogonal invariance of $\mathcal{L}(W)$.
By a similar argument using the orthogonal invariance of
$\mathcal{L}(\mathbf{A})$, we have
\begin{equation*}
\begin{split}
\mathbb{E} \|  g(\mathbf{U} \mathbf{A} \mathbf{V}' +  n^{-1/2}&W) - 
\mathbf{U}\mathbf{A}\mathbf{V}'  \|_F^2 
\\
\, & = \,
\mathbb{E} \Big[ \mathbb{E} \big(\|  g(\mathbf{U} \mathbf{A} \mathbf{V}'+n^{-1/2}W) - \mathbf{U}\mathbf{A}\mathbf{V}'  \|_F^2 \, | \, \mathbf{U}, \mathbf{V},W \big) \Big] \\
& = \,
\mathbb{E} \Big[ \mathbb{E} \big(\|  g(\mathbf{A}+n^{-1/2}W) - 
\mathbf{A}  \|_F^2 \, | \, \mathbf{U}, \mathbf{V},W \big) \Big] \\
& = \,
\mathbb{E} \|  g(\mathbf{A}+n^{-1/2}W) - \mathbf{A} \|_F^2 .
\end{split}
\end{equation*}
The final term above is $\mathbb{E} \, \mbox{Loss}(\mathbf{A},g(Y))$. 
This completes the proof.
\end{proof}

Based on Theorem \ref{thm:ortoptimality}
will restrict our attention to orthogonally 
invariant reconstruction schemes in what follows.  

As noted in introduction, 
the singular value decomposition (SVD) of the observed matrix $Y$ is 
a natural starting point for reconstruction of a signal matrix $A$.
As we show below, the SVD of $Y$ is intimately connected with orthogonally invariant reconstruction methods.
An immediate consequence of the decomposition $Y = UDV'$ is that $U' Y V = D$, so 
we can diagonalize $Y$ by means of left and right orthogonal multiplications.

The next proposition follows from our ability to diagonalize the signal matrix $A$ in the reconstruction problem.

\begin{prop}
\labelA{prop:onlySVmatter}
Let $Y = A + n^{-1/2}W$, where $W$ has an orthogonally invariant distribution.
If $g(\cdot)$ is an orthogonally invariant reconstruction scheme, then
for any fixed signal matrix $A$, the distribution of 
$\mbox{Loss}(A,g(Y))$, and in particular $\mathbb{E} \mbox{Loss}(A,g(Y))$,
depends only on the singular values of $A$.
\end{prop}

\begin{proof}
Let $U D_A V'$ be the SVD of $A$.  Then $D_A = U' A\, V$, and as 
the Frobenius norm is invariant under left and right orthogonal multiplications,
\begin{equation*}
\begin{split}
\mbox{Loss}(A,g(Y))
\, & = \,
\| \, g(Y) - A \, \|_F^2
\ = \
\| \, U' \, ( g(Y) - A ) \, V \, \|_F^2
\\ & = \,
\| \, U' g(Y) V - U'AV \, \|_F^2
\ = \
\| \, g(U' Y V) - D_A \, \|_F^2
\\ & = \,
\| \, g(D_A + n^{-1/2}U' W V) - D_A \, \|_F^2
.
\end{split}
\end{equation*}
The result now follows from the fact that $U W V'$ has the same distribution
as $W$.
\end{proof}

We now address the implications of our ability to diagonalize the observed matrix $Y$.
Let $g(\cdot)$ be an orthogonally invariant reconstruction method, and let $U D V'$
be the singular value decomposition of $Y$.  
It follows from the orthogonal invariance of $g(\cdot)$ that
\begin{equation}
\labelA{eq:preDIDO}
g(Y) 
\, = \,
g( U D V' )
\, = \,
U g(D) V' 
\, = \,
\sum_{i = 1}^m \sum_{j =1}^n c_{ij} u_i v_j'
\end{equation}
where $c_{ij}$ depend only on the singular values of $Y$.
In particular, any orthogonally invariant $g(\cdot)$ 
reconstruction method is completely 
determined by how it acts on diagonal matrices. 
The following theorem allows us to substantially refine the 
representation (\ref{eq:preDIDO}).

\begin{thm} 
\labelA{thm:DIDO}
Let $g(\cdot)$ be an orthogonally invariant reconstruction scheme. 
Then $g(Y)$ is diagonal whenever $Y$ is diagonal.
\end{thm}

\begin{proof}
Assume without loss of generality that $m \geq n$.
Let the observed matrix
$Y = \mbox{diag}(d_1,d_2,...,d_n)$,
and let $\widehat A = g(Y)$ be the reconstructed matrix. 
Fix a row index $1 \leq k \leq m$. 
We will show that $\widehat A_{kj} = 0$ for all $j \neq k$.
Let $D_L$ be an $m \times m$ matrix derived from the 
identity matrix by flipping the sign of the $k^{th}$ diagonal element.  More formally,
$
D_L
\, = \,
I - 2e_k e_k'
$,
where $e_k$ is the $k^{th}$ standard basis vector in $\reals^m$. 
The matrix $D_L$ is known as a Householder reflection. 

Let $D_R$ be the top left $n \times n$ submatrix of $D_L$. 
Clearly $D_L D_L' = I$ and $D_R D_R' = I$, so both $D_L$ and $D_R$ are orthogonal. 
Moreover, all three matrices $D_L, Y,$ and $D_R$ are diagonal, and therefore we
have the identity
$
Y
=
D_L Y D_R
.
$
It then follows from the orthogonal invariance of $g(\cdot)$ that
\[
\widehat A
\, = \,
g(Y)
\, = \,
g(D_L Y D_R)
\, = \,
D_L \,g(Y)\, D_R
\, = \,
D_L \,\widehat A \, D_R
.
\]
The $(i,j)^{th}$ element of the matrix $D_L \widehat A\, D_R$ is
$
\widehat A_{ij} (-1)^{\delta_{ik}} (-1)^{\delta_{jk}}
$,
and therefore $\widehat A_{kj} = - \widehat A_{kj}$ if $j \neq k$. 
As $k$ was arbitrary, $\widehat A$ is diagonal.
\end{proof}

As an immediate corollary of Theorem \ref{thm:DIDO} 
and equation (\ref{eq:preDIDO})
we obtain a compact, and
useful, representation of any orthogonally invariant reconstruction 
scheme $g(\cdot)$.

\begin{cor}%
\label{cor:DIDO}%
Let $g(\cdot)$ be an orthogonally invariant reconstruction method.
If the observed matrix $Y$ has singular value decomposition $Y=\sum d_j u_j v_j'$ then the
reconstructed matrix has the form
\begin{equation}
\labelA{eq:DIDO}
\widehat A 
\ = \ 
g(Y)
\, = \,
\sum_{j =1}^{m \wedge n} c_j u_j v_j' ,
\end{equation}
where the coefficients $c_j$ {\it depend only on the singular values of $Y$}.
\end{cor}

The converse of Corollary \ref{cor:DIDO} is true under a mild additional condition. 
Let $g(\cdot)$ be a reconstruction scheme such that $g(Y) = \sum c_j u_j v_j'$, where 
$c_j = c_j(d_1,\ldots,d_{m\wedge n})$ are fixed functions of the singular values of 
$Y$.  If the functions $\{ c_j(\cdot) \}$ are such that $c_i(d) = c_j(d)$ whenever $d_i = d_j$, then 
$g(\cdot)$ is orthogonally invariant.  This follows from the uniqueness of the singular value decomposition.

%

%


\section{Asymptotic Matrix Reconstruction and \\Random Matrix Theory%
\labelA{sec:asy}}


Random matrix theory is broadly concerned with the spectral 
properties of random matrices, and is an obvious starting point
for an analysis of matrix reconstruction.
The matrix reconstruction problem has several points of intersection with 
random matrix theory.  
Recently a number of authors have studied low rank deformations of Wigner matrices
\citep{capitaine2009largest, feral2007largest, maida2007large, peche2006largest}.
However, their results concern symmetric matrices, a constraint not present in the reconstruction
model, and are not directly applicable to the reconstruction
problem of interest here.  (Indeed, our simulations of non-symmetric matrices exhibit
behavior deviating from that predicted by the results of these papers.)
A signal plus noise framework similar to matrix reconstruction is studied in \cite{dozier2007empirical} and \cite{nadakuditi2007fundamental},
however both these papers model the signal matrix to be random, while in the matrix reconstruction problem we assume it to be non-random. 
\cite{karoui2008spectrum} considered the problem of estimation the eigenvalues of a population covariance matrix
from a sample covariance matrix, which is similar to the problem of estimation of the singular values of $A$ from the singular values of $Y$. However for the matrix reconstruction problem it is equally important to 
estimate the difference between the singular vectors of $A$ and $Y$, in addition to the estimate of the singular values of $A$.

Our proposed denoising scheme is based on the theory of spiked
population models in random matrix theory.  Using recent results on
spiked population models, we establish asymptotic connections between the
singular values and vectors of the signal matrix $A$ and those of 
the observed matrix $Y$.
These asymptotic connections provide us with
finite-sample estimates that can be applied in a non-asymptotic setting
to matrices of small or moderate
dimensions.


\subsection{Asymptotic Matrix Reconstruction Model}

The proposed reconstruction method is derived from an asymptotic version of the matrix 
reconstruction problem (\ref{eq:model}).  
For $n \geq 1$ let integers $m = m(n)$ be defined in such a way that 
\begin{equation}
\labelA{eq:aspect}
\frac{m}{n} \, \to \, c \, > \, 0 \ \mbox{ as } \ n \to \infty .
\end{equation}
For each $n$ let $Y$, $A$, and $W$ be $m \times n$ matrices such that 
\begin{equation}
\labelA{eq:asmodel}
	Y \,=\, A + \frac{1}{\sqrt{n}} W ,
\end{equation}
where the entries of $W$ are independent $N(0,1)$ random
variables.   
We assume that the signal matrix $A$ has 
fixed rank $r \geq 0$ and fixed non-zero singular
values $\lambda_1(A), \ldots, \lambda_r(A)$ that are 
independent of $n$.  
The constant $c$ represents the limiting aspect ratio of the observed
matrices $Y$.
The scale factor $n^{-1/2}$ ensures that the singular
values of the signal matrix are comparable to those of the noise.
We note that Model (\ref{eq:asmodel}) matches the asymptotic model used by 
\cite{capitaine2009largest,feral2007largest} in their study of fixed rank 
perturbations of Wigner matrices. 

In what follows $\lambda_j(B)$ will denote the $j$-th singular value
of a matrix $B$, and $u_j(B)$ and $v_j(B)$ will denote, respectively, 
the left and right singular values corresponding to $\lambda_j(B)$.
Our first proposition concerns the behavior of the singular values of
$Y$ when the signal matrix $A$ is equal to zero.

\begin{prop}
\labelA{prop:MP}
Under the asymptotic reconstruction model with $A = 0$ 
the empirical distribution of the singular values 
$\lambda_1(Y) \geq \cdots \geq \lambda_{m \wedge n}(Y)$ 
converges weakly to 
a (non-random) limiting distribution with density
\begin{equation}
\label{eq:MPdensity}
f_{Y}(s) 
\,=\,
\frac{s^{-1}}{\pi (c \wedge 1)}\sqrt{(a-s^2)(s^2-b)}
, 
\qquad s \in [\sqrt{a},\sqrt{b}]
,
\end{equation}
where $a = (1-\sqrt{c})^2$ and $b = (1+\sqrt{c})^2$.
Moreover, $\lambda_1(Y) \, \topr \, 1 + \sqrt{c}$ and 
$\lambda_{m \wedge n}(Y) \, \topr \, 1 - \sqrt{c}$ 
as $n$ tends to infinity.
\end{prop}

The existence and form of the density $f_Y(\cdot)$ 
are a consequence of the 
classical Mar{\v{c}}enko-Pastur theorem 
\citep{marcenko1967distribution,wachter1978strong}. 
The in-probability limits of $\lambda_1(Y)$ and
$\lambda_{m \wedge n}(Y)$ follow from 
later work of \cite{geman1980limit}
and \cite{wachter1978strong}, respectively.
If $c=1$, the density function $f_{Y}(s)$ simplifies 
to the quarter-circle law
$f_{Y}(s) = \pi^{-1} \sqrt{4-s^2}$ for $s \in [0,2]$.

The next two results concern the limiting eigenvalues and
eigenvectors of $Y$ when $A$ is non-zero.  Proposition
\ref{prop:SPM} relates the limiting eigenvalues of $Y$ to the (fixed)
eigenvalues of $A$, while Proposition \ref{prop:Paul} relates the
limiting singular vectors of $Y$ to the singular vectors of $A$.  
Proposition \ref{prop:SPM} is based on recent work of \cite{baik2006els}, 
while Proposition \ref{prop:Paul} is based on
recent work of \cite{paul2007asymptotics}, \cite{nadler2008finite}, 
and \cite{lee2010convergence}.  The proofs of both results are given in 
Section \ref{sec:proofs}.

\begin{prop}
\labelA{prop:SPM}
Let $Y$ follow the asymptotic matrix reconstruction model (\ref{eq:asmodel})
with signal singular values 
$\lambda_1(A) \geq ... \geq \lambda_r(A) > 0$. 
For $1 \leq j \leq r$, as $n$ tends to infinity,
\[
\lambda_j(Y) 
\, \topr \,
\left\{
\begin{array}{lll}
\left( 
1 + \lambda_j^2(A) + c + \frac{c}{\lambda_j^2(A)}
\right)^{1/2}
							& \mbox{ if } &  \lambda_j(A) > \sqrt[4]{c}\\
1 + \sqrt{c} 					& \mbox{ if } & 0 < \lambda_j(A) \leq \sqrt[4]{c}  \\
\end{array}
\right.
\]
The remaining singular values $\lambda_{r+1}(Y),\ldots,\lambda_{m \wedge n}(Y)$ of $Y$
are associated with the zero singular values of $A$:  their empirical
distribution converges weakly to the limiting distribution in Proposition \ref{prop:MP}.
\end{prop}

\begin{prop}
\labelA{prop:Paul}
Let $Y$ follow the asymptotic matrix reconstruction model (\ref{eq:asmodel})
with distinct signal singular values $\lambda_1(A) > \lambda_2(A) > ... > \lambda_r(A)>0$.
Fix $j$ such that $\lambda_j(A) > \sqrt[4]{c}$. Then as $n$ tends to infinity,
\[
\big\langle u_j(Y),u_j(A) \big\rangle^2
\, \topr \,
\Big(1-\frac{c}{\lambda_j^4(A)}\Big) \,\big/\, \Big( 1 + \frac{c}{\lambda_j^2(A)}\Big)
\]
and
\[
\big\langle v_j(Y),v_j(A) \big\rangle^2
\, \topr \,
\Big(1-\frac{c}{\lambda_j^4(A)}\Big) \,\big/\, \Big( 1 + \frac{1}{\lambda_j^2(A)}\Big)
\]
Moreover, if $k=1,\ldots,r$ not equal to $j$ then
$\langle u_j(Y), u_k(A) \rangle \topr 0$ and  
$\langle v_j(Y), v_k(A) \rangle \topr 0$
as $n$ tends to infinity.
\end{prop}

The limits established in Proposition \ref{prop:SPM} 
indicate a phase transition.  If the singular value $\lambda_j(A)$ 
is less than or equal to $\sqrt[4]{c}$ then, asymptotically, the
singular value $\lambda_j(Y)$
lies within the support of the Mar{\v{c}}enko-Pastur distribution
and is not distinguishable from the noise singular values. 
On the other hand, if $\lambda_j(A)$ exceeds $\sqrt[4]{c}$ then,
asymptotically, $\lambda_j(Y)$ lies outside the support of the Mar{\v{c}}enko-Pastur distribution,
and the corresponding left and right singular vectors of $Y$ are associated with
those of~$A$ (Proposition \ref{prop:Paul}).


\section{Proposed Reconstruction Method%
\labelA{sec:RMTmethod}}

Assume for the moment
that the variance $\sigma^2$ of the noise is known, 
and equal to one.  
Let $Y$ be an observed $m \times n$ matrix generated from
the additive model $Y = A + n^{-1/2} W$, and let 
\[
Y \,=\, \sum_{j=1}^{m \wedge n} \lambda_j(Y) \, u_j(Y) v_j'(Y)
\]
be the SVD of $Y$.  Following the discussion in Section \ref{sec:invariant}, we
seek an estimate $\widehat{A}$ of the signal matrix $A$ having the form
\[
\widehat A \,=\, \sum_{j=1}^{m \wedge n} c_j \, u_j(Y) v_j'(Y) ,
\]
where each coefficient $c_j$ depends only on the singular values
$\lambda_1(Y), \ldots,$ $\lambda_{m \wedge n}(Y)$ of $Y.$  
We derive $\widehat A$ from the 
limiting relations in Propositions \ref{prop:SPM} 
and \ref{prop:Paul}.   By way of approximation, we treat these
relations as exact in the non-asymptotic setting under study, 
using the symbols $\eqa$, $\leqa$ and $\ga$ to denote limiting 
equality and inequality relations. 

Suppose initially that the singular values and vectors
of the signal matrix $A$ are known.  In this case we wish to find 
coefficients $\{c_j\}$ minimizing  
\[
\mbox{Loss}(A,\widehat A)
\, = \, 
\big\| 
\sum_{j=1}^{m \wedge n} c_j \, u_j(Y) v_j'(Y)
\, - \, 
\sum_{j=1}^{r} \lambda_j(A) \, u_j(A) v_j'(A)
\big\|_F^2 .
\]
Proposition \ref{prop:SPM} shows that asymptotically the information about the singular values of $A$ that are smaller that $\sqrt[4]{c}$ is not recoverable from the singular values of $Y$.
Thus we can restrict the first sum to the first $r_0 = \#\{j:\lambda_j(A) > \sqrt[4]{c}\}$ terms
\[
\mbox{Loss}(A,\widehat A)
\, = \, 
\big\| 
\sum_{j=1}^{r_0} c_j \, u_j(Y) v_j'(Y)
\, - \, 
\sum_{j=1}^{r} \lambda_j(A) \, u_j(A) v_j'(A)
\big\|_F^2.
\]
Proposition \ref{prop:Paul} ensures that the left singular vectors $u_j(Y)$ and $u_k(A)$ are
asymptotically orthogonal for $k=1,\ldots,r$ not equal to $j=1,\ldots,r_0$, and therefore
\[
\mbox{Loss}(A,\widehat A)
\, \eqa \, 
\sum_{j=1}^{r_0} 
\big\| 
c_j \, u_j(Y) v_j'(Y)
\, - \, 
\lambda_j(A) \, u_j(A) v_j'(A)
\big\|_F^2 
\,+\,
\sum_{j=r_0+1}^{r} 
\lambda^2_j(A)
.
\]
Fix $1 \leq j \leq r_0$.  Expanding the $j$-th term in the 
above sum gives
\begin{equation*}
\begin{split}
\big\| \lambda_j(A) \, u_j(A) & v_j'(A) \ - \ c_j \, u_j(Y)  v_j'(Y) \big\|_F^2 
\\[0.05in] & = \,
c_j^2 \, \big\|  u_j(Y) v_j'(Y) \big\|_F^2 
\  + \ 
\lambda_j^2(A) \, \big\| u_j(A) v_j'(A) \big\|_F^2
\\[0.05in] & \phantom{ = \,} \ - \ 
2 c_j \lambda_j(A) \,  \big\langle u_j(A) v_j'(A),  u_j(Y) v_j'(Y) \big\rangle
\\[0.05in] & = \,
\lambda_j^2(A) 
\  + \ 
c_j^2  
\ - \ 
2 c_j \lambda_j(A) \,  
\big\langle u_j(A), u_j(Y) \big\rangle \,
\big\langle v_j(A), v_j(Y) \big\rangle . \\
\end{split}
\end{equation*}
Differentiating the last expression with respect to $c_j$ yields the optimal value
\begin{equation}
\label{eq:optcoeff}
c_j^*  \, = \,
\lambda_j(A) \,
\big\langle u_j(A), u_j(Y) \big\rangle \,
\big\langle v_j(A), v_j(Y) \big\rangle
.
\end{equation}

In order to estimate the coefficient $c_j^*$ we consider
separately singular values of $Y$ that are at most, or 
greater than $1 + \sqrt{c}$, where $c = m/n$ is the aspect
ratio of $Y$.
By Proposition \ref{prop:SPM}, the asymptotic relation
$\lambda_j(Y) \leqa 1 + \sqrt{c}$ implies 
$\lambda_j(A) \leq \sqrt[4]{c}$, and 
in this case the $j$-th singular value of $A$ is not recoverable from $Y$.
Thus if $\lambda_j(Y) \leq 1 + \sqrt{c}$ we set the corresponding coefficient $c_j^* = 0$. 

On the other hand, the asymptotic relation
$\lambda_j(Y) \ga 1 + \sqrt{c}$ implies that $\lambda_j(A) > \sqrt[4]{c}$, and that
each of the inner products in (\ref{eq:optcoeff}) are
asymptotically positive.  
The displayed equations in 
Propositions \ref{prop:SPM} and \ref{prop:Paul} can then be used
to obtain estimates of each term in (\ref{eq:optcoeff}) based only
on the (observed) singular values of $Y$ and its aspect ratio $c$.  
These equations yield the following relations:
\[
\widehat \lambda_j^2(A)
\ = \
\frac{1}{2}
\bigg[\lambda_j^2(Y) - (1+c)  + \sqrt{ [\lambda_j^2(Y) - (1+c)]^2 - 4c } \bigg]
\ \mbox{ estimates } \ \lambda_j^2(A) ,
\]
\[
\hat \theta_j^2
\, = \,  
\left(1-\frac{c}{\widehat \lambda_j^4(A)} \right)
\,\big/\, \left( 1 + \frac{c}{\widehat \lambda_j^2(A)} \right)
\ \mbox{ estimates } \ \langle u_j(A), u_j(Y) \rangle^2,
\]

\[
\hat \phi_j^2
\, = \,  
\left(1-\frac{c}{\widehat \lambda_j^4(A)} \right)
\,\big/\, \left( 1 + \frac{1}{\widehat \lambda_j^2(A)} \right)
\ \mbox{ estimates } \ \langle v_j(A), v_j(Y) \rangle^2.
\]
\vskip.1in
\noindent
With these estimates in hand, the proposed reconstruction 
scheme is defined via the equation
\begin{equation}
\labelA{eq:RMT-KV}
G_o^{RMT}(Y)
\, = \,
\sum_{\lambda_j(Y) > 1 + \sqrt{c}} 
\widehat \lambda_j(A) \, \hat \theta_j \, \hat \phi_j \,
u_j(Y)  v_j'(Y),
\end{equation}
where $\widehat \lambda_j(A)$, $\hat \theta_j$, and $\hat \phi_j$ 
are the positive square roots of the estimates defined above.

\vskip.1in

The RMT method shares features with both hard and soft thresholding.
It sets to zero singular values of $Y$ 
smaller than the threshold ($1+\sqrt{c}$), and it shrinks the remaining singular values towards zero. 
However, unlike soft thresholding the amount of shrinkage depends on the singular values,
the larger singular values are shrunk less than the smaller ones.
This latter feature is similar to that of LASSO type estimators based on an $L_q$ penalty  with $0 < q <1$ (also known as bridge estimators \citealp{fu1998penalized}).
It is important to note that, unlike hard and soft thresholding schemes, the 
proposed RMT method has no tuning parameters.  The only unknown, the noise 
variance, is estimated within the procedure.

\vskip 0.1in

In the general version of the matrix reconstruction problem, 
the variance $\sigma^2$ of the noise is not known.  In this case, 
given an estimate $\widehat \sigma^2$ of $\sigma^2$, such
as that described below, we may define
\begin{equation}
\labelA{eq:RMT-EV}
G^{RMT}(Y)
\, = \,
\widehat \sigma \ G_o^{RMT} \left( \frac{Y}{\widehat \sigma} \right) ,
\end{equation}
where $G_o^{RMT}(\cdot)$ is the estimate defined in 
(\ref{eq:RMT-KV}).

\subsection{Estimation of the Noise Variance%
\labelA{sec:noisevar}}

Let $Y$ be derived from the asymptotic reconstruction model
$Y \, = \, A + \sigma n^{-1/2} W$ with sigma unknown.
While it is natural to try to estimate $\sigma$ from the
entries of $Y$, the following general results indicate that,
under mild conditions, it is sufficient to consider estimates
based on the singular values of $Y$.  The results and their
proofs parallel those in Section \ref{}.

\begin{definition}
A function $s(\cdot) : \reals^{m\times n} \to \reals$ is orthogonally invariant 
if for any $m \times n$ matrix $Y$ and any orthogonal matrices $U$ and 
$V$ of appropriate sizes, $s(Y)=s(UYV')$. 
\end{definition}

\begin{prop}
\label{prop:sigma1}
A function $s(\cdot) : \reals^{m\times n} \to \reals$ is orthogonally invariant if and only if 
$s(Y)$ depends only on the singular values of $Y$.
\end{prop}

\begin{prop}
\label{prop:sigma2}
Let $s(\cdot) : \reals^{m\times n} \to \reals$.
Then there is an orthogonally invariant function $\tilde s(\cdot)$ with the following property.
Let $\mathbf{A}$ and $W$ be independent $m \times n$ random matrices with orthogonally invariant distributions, and let $Y = \mathbf{A} + \sigma n^{-1/2}W$ for some $\sigma$. 
Then $\tilde s(Y)$ has the same expected value as $s(Y)$ and a smaller or equal variance.
\end{prop}

Based Propositions \ref{prop:sigma1} and \ref{prop:sigma2} we restrict our attention to the estimates of $\sigma$ that depend only on the singular values of $Y$.
It follows from Proposition \ref{prop:SPM} that the empirical distribution of the 
$(m - r)$ singular values $S = \{ \lambda_j(Y/\sigma) : \lambda_j(A) = 0 \}$ 
converges weakly to a distribution with density (\ref{eq:MPdensity})
supported on the interval $[|1 - \sqrt{c}|, 1 + \sqrt{c}]$.  
Following the general approach outlined in \cite{gyorfi1996minimum},
we estimate $\sigma$ by minimizing the Kolmogorov-Smirnov
distance 
between 
the observed sample distribution of the singular values of $Y$ and that predicted by theory.
Let $F$ be the CDF of the density (\ref{eq:MPdensity}). 
For each $\sigma > 0$ let $\widehat S_\sigma$ be the set of 
singular values $\lambda_j(Y)$
that fall in the interval $[\sigma |1 - \sqrt{c}|, \sigma (1 + \sqrt{c})]$, and let
$\widehat F_\sigma$ be the empirical CDF of $\widehat S_\sigma$.
Then
\[
K(\sigma) \,=\, \sup_s |F(s/\sigma) - \widehat F_\sigma(s)|
\]
is the Kolmogorov-Smirnov distance between the empirical
and theoretical singular value distribution functions, and we define
\begin{equation}
\label{sigmahat}
\hat \sigma(Y) \,=\, \mathop{\arg\min}_{\sigma > 0} K(\sigma)
\end{equation}
to be the value of $\sigma$ minimizing $K(\sigma)$.  
A routine argument shows that the estimator $\hat \sigma$ is 
scale invariant, in the sense that
$\hat \sigma(\beta \, Y) = \beta \, \hat \sigma(Y)$
for each $\beta > 0$.

By considering the jump points of the empirical CDF $\widehat F_\sigma(s)$,
the supremum in $K(\sigma)$ simplifies to
\[
K(\sigma) \,=\, \max_{s_i \in \widehat S_\sigma} \left| F(s_i/\sigma) - \frac{i-1/2}{|\widehat S_\sigma|} \right| + \frac{1}{2|\widehat S_\sigma|}
,
\]
where $\{s_i\}$ are the ordered elements of $\widehat S_\sigma$.
The objective function $K(\sigma)$ is discontinuous at points where the $\widehat S_\sigma$ changes, so we minimize it over a fine grid of points $\sigma$ in the range where $|\widehat S_\sigma| > (m \wedge n)/2$ and $\sigma (1 + \sqrt{c}) < 2 \lambda_1(Y)$. The closed form of the cumulative distribution function $F(\cdot)$ is presented in Section \ref{sec:cdf}.

\section{Simulations%
\labelA{sec:simulations}}

We carried out a simulation study to evaluate the performance 
of the RMT reconstruction 
scheme $G^{RMT}(\cdot)$ defined in (\ref{eq:RMT-EV}) using the variance
estimate $\widehat \sigma$ in (\ref{sigmahat}).
The study compared the performance of $G^{RMT}(\cdot)$ to 
three alternatives: the best hard thresholding reconstruction scheme, 
the best soft thresholding reconstruction scheme, and the best orthogonally 
invariant reconstruction scheme.  Each of the three competing alternatives is 
an oracle-type procedure that is based on information about the signal
matrix $A$ that is not available to $G^{RMT}(\cdot)$.

\subsection{Hard and Soft Thresholding Oracle Procedures%
\labelA{sec:simparams}}
Hard and soft thresholding schemes require specification of a 
threshold parameter that can depend on the observed matrix $Y$.
Estimation of the noise variance can be incorporated into the 
choice of the threshold parameter.  In order to compare the 
performance of $G^{RMT}(\cdot)$ against every possible hard
and soft thresholding scheme, we define  oracle 
procedures
\begin{equation}
\label{eq:HSO}
G^{H}(Y)
\, = \,
g_{\lambda^*}^{H} (Y ) 
\ \mbox{ where } \ 
\lambda^*
\ = \ 
\mathop{\arg\min}_{\lambda > 0} \big\| A - g_\lambda^H (Y) \big\|_F^2
\end{equation}
\begin{equation}
\label{eq:HSO}
G^{S}(Y)
\, = \,
g_{\nu^*}^{S} (Y ) 
\ \mbox{ where } \ 
\nu^*
\ = \ 
\mathop{\arg\min}_{\nu > 0} \big\| A - g_\nu^S (Y) \big\|_F^2
\end{equation}
that make use of the signal $A$.  By definition, the loss
$\| A - G^{H}(Y) \|_F^2$ of $G^{H}(Y)$ is less than that of any hard thresholding 
scheme, and similarly the loss of $G^{S}(Y)$ is less than that of any
soft thresholding procedure.  In effect, the oracle procedures have access to 
both the unknown signal matrix $A$ and the unknown variance $\sigma$.  They
are constrained only by the form of their respective thresholding families.
The oracle procedures are not realizable in practice.


\subsection{Orthogonally Invariant Oracle Procedure%
\labelA{sec:oracle}}


As shown in Corrolary~\ref{cor:DIDO}, every orthogonally invariant reconstruction scheme 
$g(\cdot)$ has the form
\[
g(Y)
\, = \,
\sum_{j =1}^{m \wedge n} c_j \, u_j(Y) v_j(Y)'
,
\]
where the coefficients $c_j$ are functions of the singular values of $Y$.  

The orthogonally invariant oracle scheme has coefficients $c_j^o$ 
minimizing the loss
\[
\big\| A - \sum_{j =1}^{m \wedge n} c_j \, u_j(Y) v_j(Y)' \big\|_F^2
\]
over all choices $c_j$.  As in the case with the hard and soft thresholding oracle schemes,
the coefficients $c_j^o$ depend on the signal matrix $A$, which in practice is unknown.

The (rank one) matrices $\{u_j(Y) v_j(Y)'\}$ form an orthonormal basis of 
an $m \wedge n$-dimensional subspace of the $mn$-dimensional
space of all $m \times n$ matrices.  Thus the optimal coefficient $c_j^o$ is
simply the matrix inner product $\langle A, u_j(Y) v_j(Y)' \rangle$, and the
orthogonally invariant oracle scheme has the form of a projection
\begin{equation}
\labelA{eq:OIoracle}
G^{*}(Y)
\, = \,
\sum_{j =1}^{m \wedge n}  \big\langle A, u_j(Y) v_j(Y)' \big\rangle \, u_j(Y) v_j(Y)'.
\end{equation}
By definition, for any orthogonally invariant reconstruction
scheme $g(\cdot)$ and observed matrix $Y$, we have
$\| A - G^{*}(Y) \|_F^2 \leq \| A - g(Y) \|_F^2$.


\subsection{Simulations}
We compared the reconstruction schemes 
$G^{H}(Y), \ G^{S}(Y),$ and 
$G^{RMT}(Y)$ to $G^{*}(Y)$ on a wide variety of signal matrices 
generated according
to the model (\ref{eq:model}).  As shown in Proposition \ref{prop:onlySVmatter}, 
the distribution of the loss $\| A - G(Y) \|_F^2$ depends only on the singular 
values of $A$, so we considered only diagonal signal matrices.
As the variance estimate used in $G^{RMT}(\cdot)$ is scale invariant,
all simulations were run with noise of unit variance.  (Estimation of noise
variance is not necessary for the oracle reconstruction schemes.)

\subsubsection{Square Matrices}

Our initial simulations considered $1000 \times 1000$ square matrices. 
Signal matrices $A$ were generated using three parameters: the rank $r$;
the largest singular value $\lambda_1(A)$; and the decay profile
of the remaining
singular values.  We considered ranks $r \in \{1, 3, 10, 32, 100\}$
corresponding to successive powers of $\sqrt{10}$ up to $(m \wedge n)/10$, and 
maximum singular values $\lambda_1(A) \in \{0.9,1,1.1,...,10\}\sqrt[4]{c}$
falling below and above the critical threshold of $\sqrt[4]{c} = 1$.
We considered several coefficient decay profiles: (i) all coefficients equal;
(ii) linear decay to zero; (iii) linear decay to $\lambda_1(A)/2$; and
(iv) exponential decay as powers of 0.5, 0.7, 0.9, 0.95, or 0.99.
Independent noise matrices $W$ were generated for each signal matrix $A$. 
All reconstruction schemes were then applied to the resulting matrix $Y = A + n^{-1/2} W$.
The total number of generated signal matrices was 3,680.


\begin{figure}[htbp]
\begin{center}
\includegraphics[width=6cm]{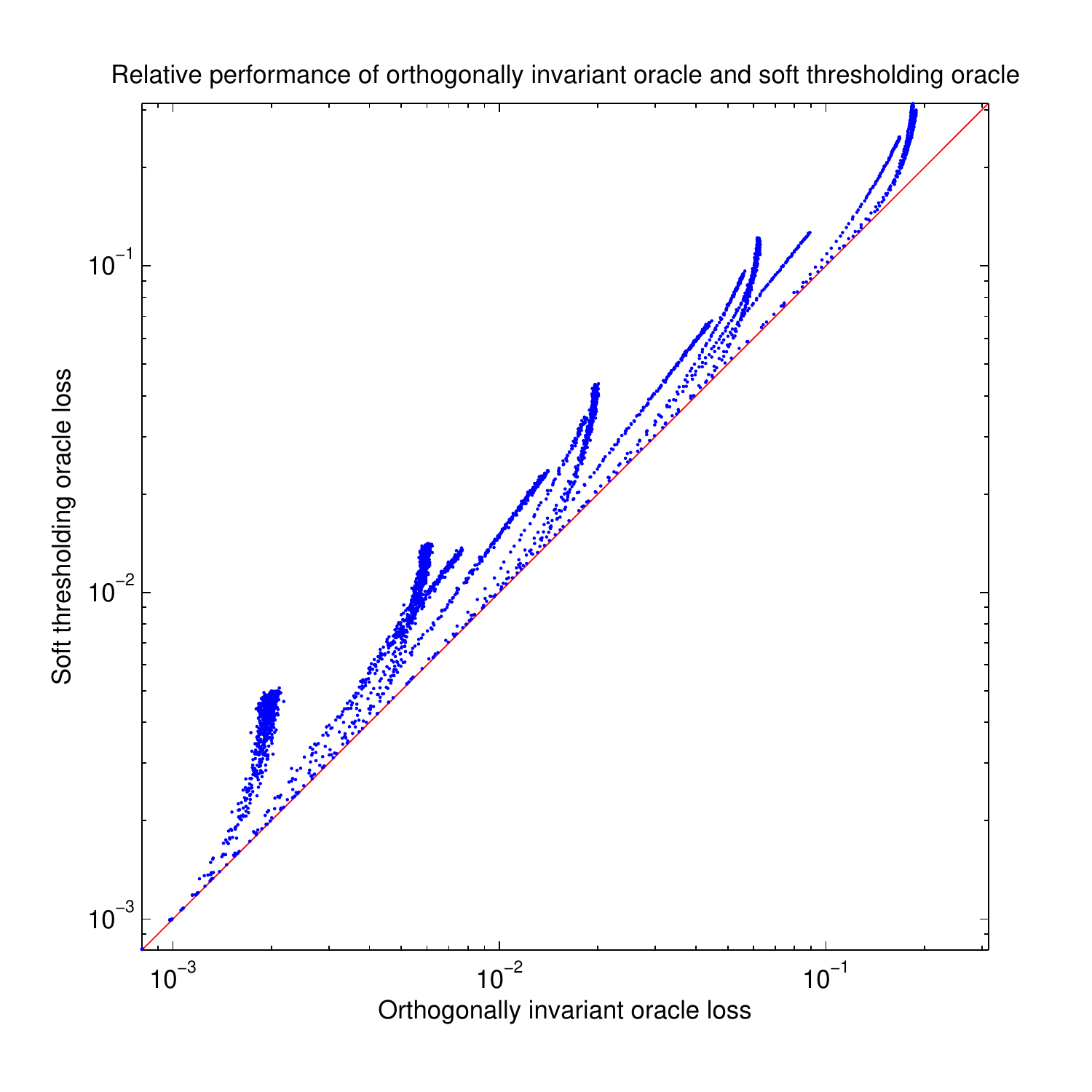}%
~
\includegraphics[width=6cm]{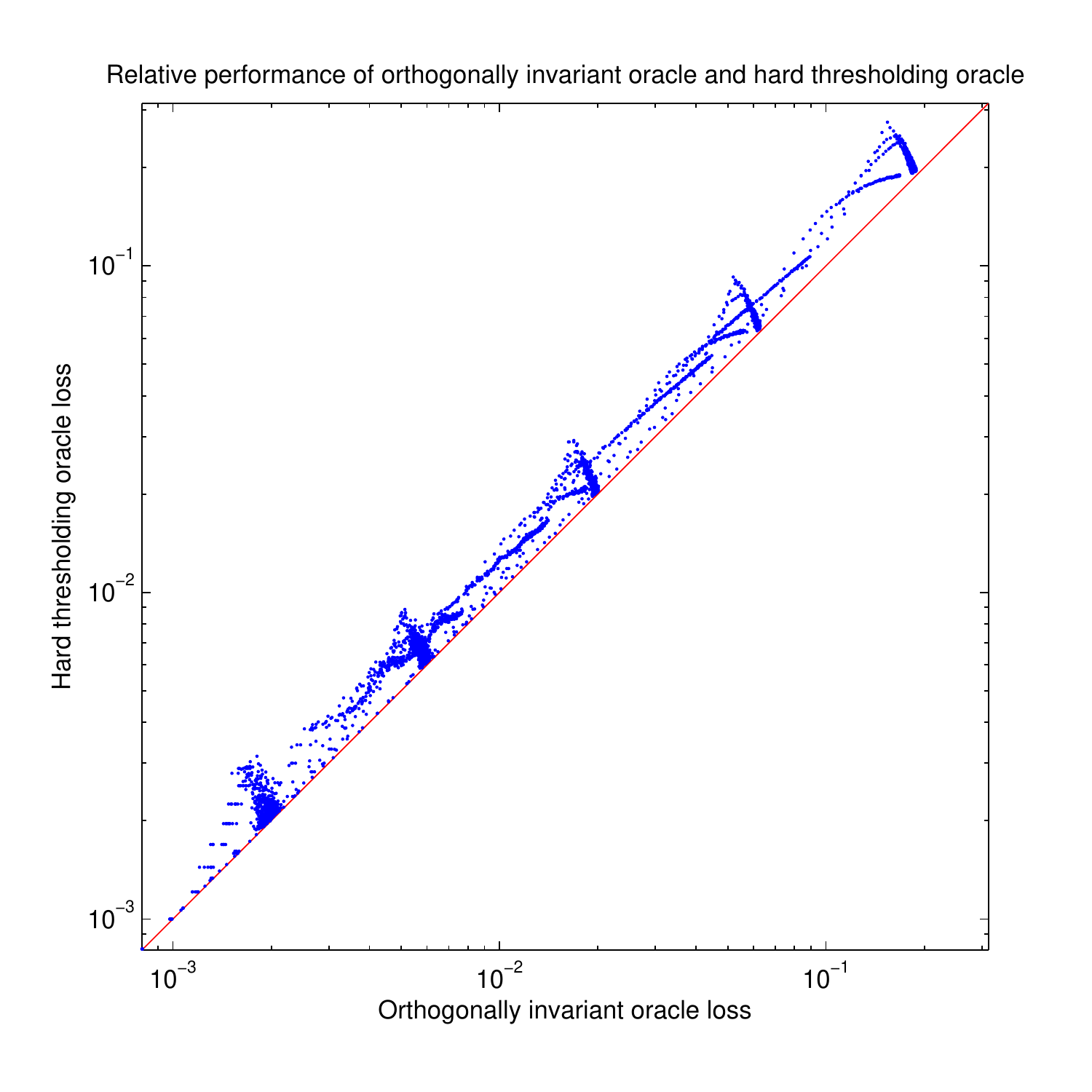}%
\caption{Relative performance of soft and hard thresholding  method against the orthogonally invariant 
oracle for $1000 \times 1000$ matrices.}
\label{fig:3plots1}
\label{fig:3plots2}
\end{center}
\end{figure}

\comment{
\begin{figure}[htbp]
\begin{center}
\includegraphics[width=4cm]{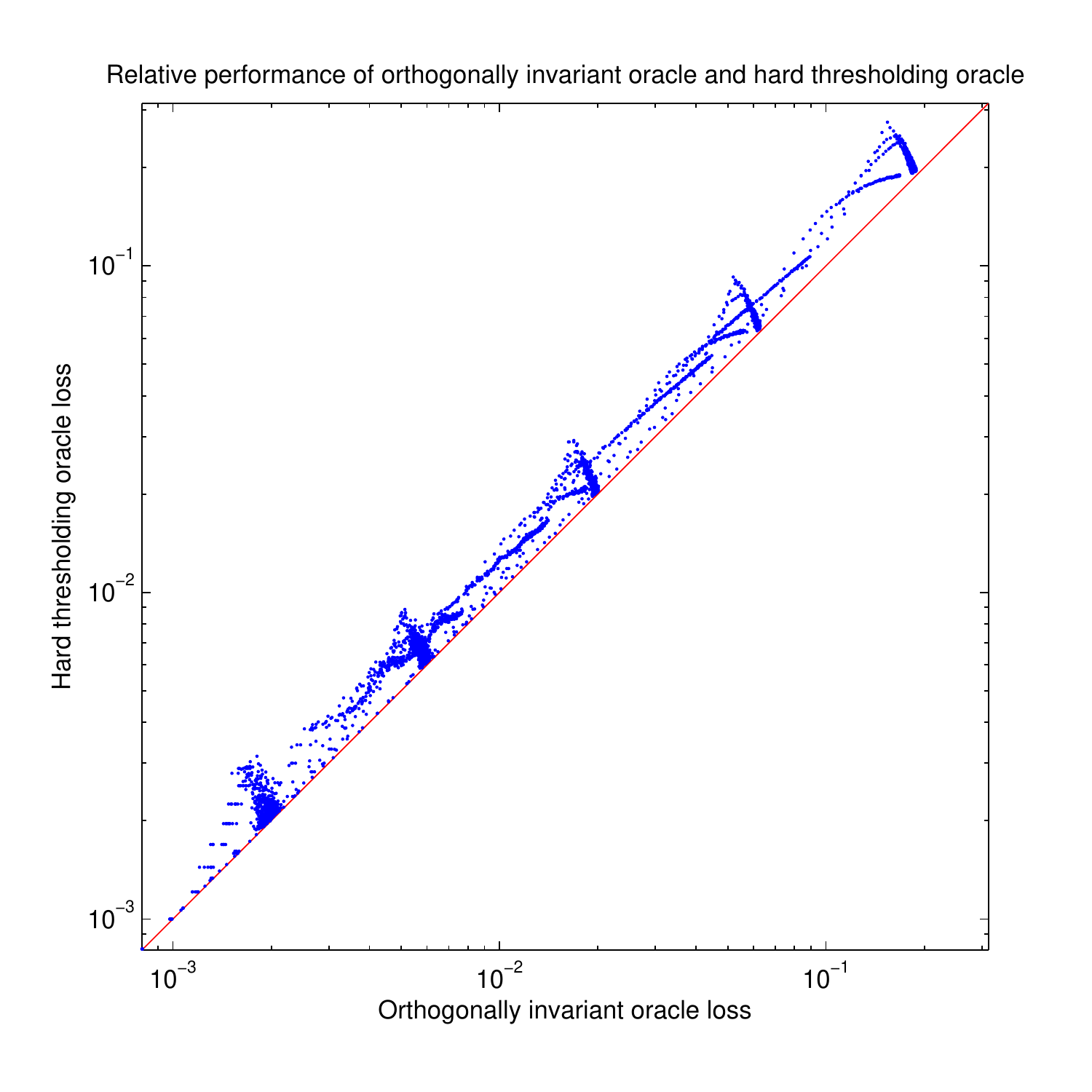}%
\caption[Relative performance of hard  thresholding and OI oracle methods.]%
{Relative performance of hard thresholding and orthogonally invariant oracle for $1000 \times 1000$ matrices.}
\label{fig:3plots2}
\end{center}
\end{figure}
}

\begin{figure}[htbp]
\begin{center}
\includegraphics[width=6cm]{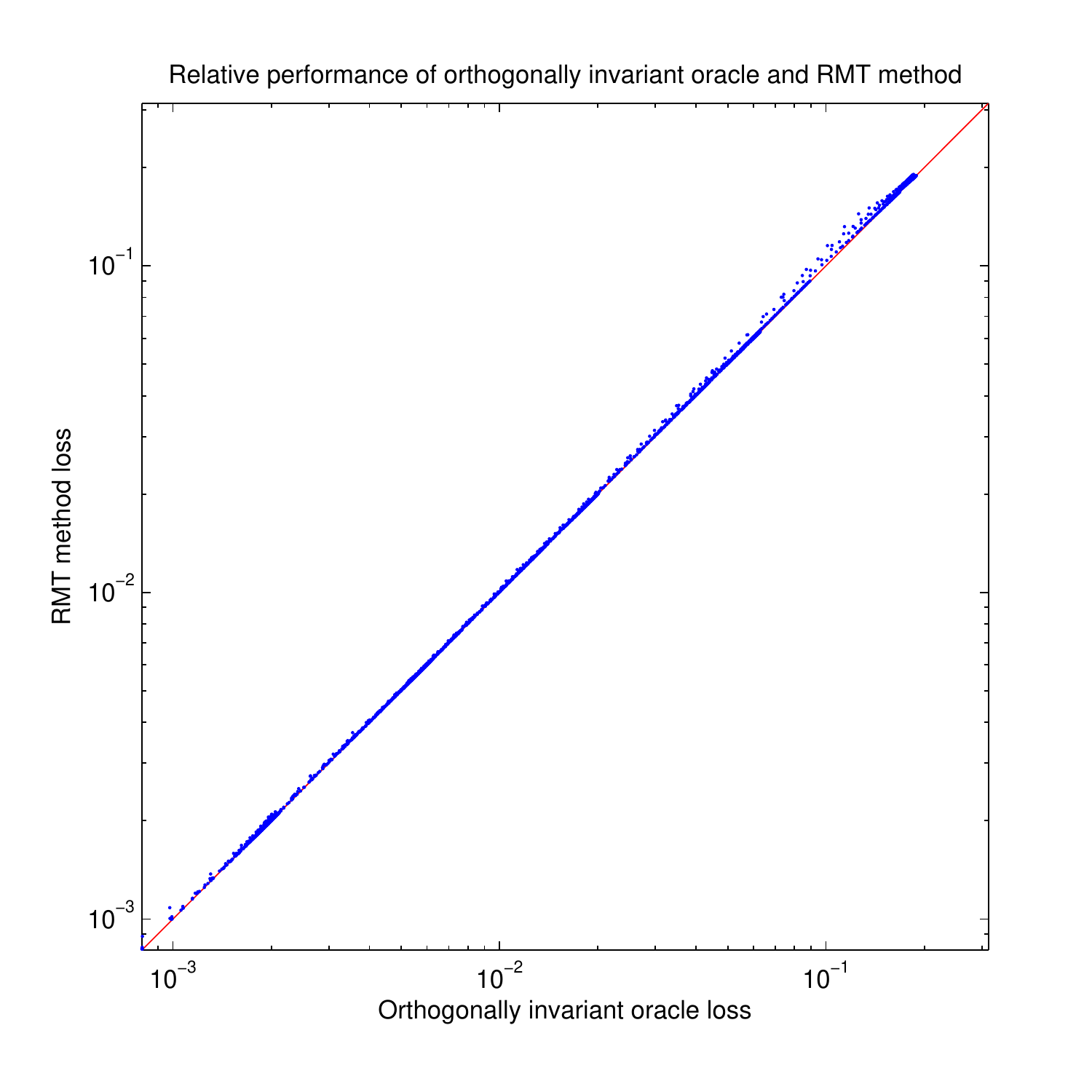}%
\caption[Relative performance of RMT method and OI oracle.]%
{Relative performance of RMT method and orthogonally invariant oracle method for $1000 \times 1000$ matrices.}
\label{fig:3plots3}
\end{center}
\end{figure}

Figures \ref{fig:3plots1} and \ref{fig:3plots3} illustrate,
respectively, the loss 
of the best soft thresholding, best hard thresholding and RMT reconstruction methods
(y axis) relative to the best orthogonally invariant scheme (x axis).
In each case the diagonal represents the performance of the orthogonally
invariant oracle: points farther from the diagonal represent worse performance.
The plots show clearly that $G^{RMT}(\cdot)$
outperforms the oracle schemes $G^{H}()$ and $G^{S}(\cdot)$, and 
has performance comparable to that of the orthogonally invariant oracle.
In particular, $G^{RMT}(\cdot)$
outperforms any hard or soft thresholding scheme, even if the latter schemes 
have access to the unknown variance $\sigma$ and the signal matrix $A$.

In order to summarize the results of our simulations, for each
scheme $G(\cdot)$ and for each matrix $Y$ 
generated from a signal matrix $A$ we calculated the relative excess 
loss of $G()$ with respect to $G^*()$:
\begin{equation}
\labelA{eq:REL}
\mbox{REL}(A, G(Y)) 
\, = \,
\frac{\mbox{Loss}( A, G(Y) )}{\mbox{Loss}( A, G^*(Y) )} - 1
\end{equation}
The definition of $G^*()$ ensures that relative excess loss is non-negative.  The
average REL of $G^{S}(\cdot)$, $G^{H}()$, and $G^{RMT}(\cdot)$ across the
3680 simulated $1000 \times 1000$ matrices was 68.3\%, 18.3\%, and 0.61\% respectively. 
Table~\ref{tbl:RELsquare} summarizes these results, and the results of analogous
simulations carried out on square matrices of different dimensions.
The table clearly shows the strong performance of RMT method for matrices with 
at least $50$ rows and columns. Even for $m=n=50$, the average relative excess loss of the 
RMT method is almost twice smaller then those of the oracle soft and hard thresholding methods. 

\vskip.1in

\begin{table}[htbp]
\begin{center}
\caption[AREL of different methods for square matrices.]%
{Average relative excess losses of oracle soft thresholding, oracle hard thresholding 
and the proposed RMT reconstruction method for square matrices of different dimensions.
\labelA{tbl:RELsquare}
}
\begin{tabular}{|c|r|r|r|r|r|r|r|r|r|r|r|r|} \hline
 \multicolumn{2}{|r|}{Matrix size (square)} &  2000 &  1000 &   500 &   100 &    50 \\  \hline  \hline
 & $G^{{}_S}(\cdot)$ &  0.740 & 0.683 & 0.694 & 0.611 & 0.640 \\ \cline{2-7}
\ Scheme \
 & $G^{{}_H}(\cdot)$ & 0.182 & 0.183 & 0.178 & 0.179 & 0.176 \\  \cline{2-7}
 &  $G^{{}_{RMT}}(\cdot)$ & 0.003 & 0.006 & 0.008 & 0.029 & 0.071 \\ \hline
\end{tabular}%
\end{center}
\end{table}

\subsubsection{Rectangular Matrices}

We performed simulations for rectangular matrices of different dimensions
$m,n$ and different aspect ratios $c = m/n$.   
For each choice of dimensions $m,n$ we simulated target matrices using the same rules as in the square case:
rank $r \in \{1, 3, 10, 32, \ldots\}$ not exceeding $(m \wedge n) / 10$,
maximum singular values $\lambda_1(A) \in \{0.9,1,1.1,...,10\} \sqrt[4]{c}$,
and coefficients decay profiles like those above.
A summary of the results is given in Table \ref{tbl:RELaspect}, which shows the average REL for matrices with 2000 rows and 10 to 2000 columns. Although random matrix theory used to construct the RMT scheme requires $m$ and $n$ to tend to infinity and at the same rate, the numbers in Table \ref{tbl:RELaspect} clearly show that the performance of the RMT scheme is excellent even for small $n$, where average REL ranges between 0.3\% and 0.54\%. 
The average REL of soft and hard thresholding are above 18\% in each of the simulations.

\begin{table}[htbp]
\begin{center}
\caption[AREL of different methods for rectangular matrices.]
{Average relative excess loss of oracle soft thresholding, oracle hard thresholding, 
and RMT reconstruction schemes for matrices with different dimensions
and aspect ratios.
\labelA{tbl:RELaspect}
}
\begin{tabular}{|c|r|r|r|r|r|r|r|r|r|r|r|r|} \hline
Matrix  &     $m$ &  2000 &  2000 &  2000 &  2000 &  2000 &   2000 \\ \cline{2-8}
size  &     $n$ &  2000 &  1000 &   500 &   100 &    50 &     10 \\ \hline \hline
&  $G^{{}_S}(\cdot)$ & 0.740 & 0.686 & 0.653 & 0.442 & 0.391 & 0.243 \\ \cline{2-8}
Method 
 &  $G^{{}_H}(\cdot)$ & 0.182 & 0.188 & 0.198 & 0.263 & 0.292 & 0.379 \\ \cline{2-8}
 &  $G^{{}_{RMT}}(\cdot)$ &  0.003 & 0.004 & 0.004 & 0.004 & 0.004 & 0.005 \\ \hline
\end{tabular}
\end{center}
\end{table}

\section*{Acknowledgments}

This work was supported, in part, by grants from US EPA (RD832720 and
RD833825) and grant from NSF (DMS-0907177).

The authors would like to thank Florentina Bunea and Marten Wegkamp for
helpful discussions.

\section*{Appendix}%
\subsection{Cumulative Distribution Function for Variance \\Estimation%
\labelA{sec:cdf}}


The  cumulative density function $F(\cdot)$ is calculated as the integral of $f_{n^{-1/2}W}(s)$.
For $c=1$ ($a=0, b=4$) it is a common integral 
\begin{equation*}
	F(x)
 \,=\, 
\int_{\sqrt{a}}^x f(s) ds 
\,=\,
\frac{1}{\pi} \int_{0}^x \sqrt{b-s^2} ds 
\,=\,
\frac{1}{2\pi}
\left(
x \sqrt{4-x^2} + 4 \arcsin \frac{x}{2}
\right)
\end{equation*}
For $c\neq 1$ the calculations are more complicated. First we perform the change of variables $t = s^2$, which yields
\begin{equation*}
\begin{split}
F(x)
 \, & =\, 
\int_{\sqrt{a}}^x f(s) ds 
\,=\,
C \int_{\sqrt{a}}^x s^{-2} \sqrt{(b-s^2)(s^2-a)} ds^2 
\\ & = \,
C \int_{a}^{x^2} t^{-1} \sqrt{(b-t)(t-a)}dt
,
\end{split}
\end{equation*}
where $C = 1/(2\pi (c \wedge 1))$.

Next we perform a change of variables $y = t - [a+b]/2$ to make the expression in the square root look like $h^2-x^2$, giving  
\begin{equation*}
\begin{split}
F(x)
 \, & =\, 
C \int_{-[b-a]/2}^{x^2-[a+b]/2}
\frac{\sqrt{([b-a]/2 - y)(y + [b-a]/2)}}{y +[a+b]/2} dy 
\\ & =\,
C \int_{-2\sqrt{c}}^{x^2-(1+c)}
\frac{\sqrt{4c - y^2}}{y + 1 + c} dy 
,
\end{split}
\end{equation*}
The second equality above uses the fact that $a+b = 2(1+c)$ and $b-a = 4\sqrt{c}$. 
The simple change of variables $y = 2\sqrt{c}z$ is performed next to make the numerator $\sqrt{1-z^2}$:
\[
F(x)
 \,=\, 
\frac{\sqrt{c}}{\pi (c \wedge 1)} \int_{-1}^{[x^2-(1+c)]/2\sqrt{c}}
\frac{\sqrt{1 - z^2}}{z + (1 + c)/2\sqrt{c}} dz 
\]
Next, the formula
\begin{equation*}
\begin{split}
\int \frac{\sqrt{1-z^2}}{z+q} dw
\, = &\, 
\sqrt{1-z^2} + q \arcsin(z)
\\ & 
- \sqrt{q^2-1} \arctan 
\left[
\frac{q z+1}{\sqrt{(q^2-1)(1-z^2)}} 
\right]
\end{split}
\end{equation*}
is applied to find the closed form of $F(x)$ by substituting $z = [x^2-(1+c)]/2\sqrt{c}$ and $q = (1 + c)/2\sqrt{c}$.
The final expression above can be simplified as
$\sqrt{q^2-1} = \sqrt{[(1 + c)/2\sqrt{c}]^2 - 1} = |1-c|/2\sqrt{c}$.

\subsection{Limit Theorems for Asymptotic Matrix \\Reconstruction Problem%
\labelA{sec:proofs}}

Propositions \ref{prop:SPM} and \ref{prop:Paul} in Section \ref{sec:asy} provide an
asymptotic connection
between the eigenvalues and eigenvectors of the signal 
matrix $A$ and those of the observed matrix $Y$.
Each proposition is derived from recent work in random matrix theory 
on spiked population models.
Spiked population models were introduced by \cite{johnstone2001distribution}.

\subsubsection{The Spiked Population Model}

The spiked population model is formally defined as follows.
Let $r \geq 1$ and constants $\tau_1 \geq \cdots \geq \tau_r > 1$
be given, and for $n \geq 1$ let integers $m = m(n)$ be defined 
in such a way that 
\begin{equation}
\label{eq:aspectspm}
\frac{m}{n} \, \to \, c \, > \, 0 \ \mbox{ as } \ n \to \infty .
\end{equation}
For each $n$ let 
\[
T = \mbox{diag}(\tau_1,\ldots,\tau_r, 1, \ldots, 1)
\]
be an $m \times m$ diagonal matrix (with $m = m(n)$), and let
$X$ be an $m \times n$ matrix 
with independent $N_m(0,T)$ columns.  Let $\widehat T = n^{-1} X X'$
be the sample covariance matrix of $X$.


The matrix $X$ appearing in the spiked population model may be
decomposed as a sum of matrices that parallel those in the
matrix reconstruction problem.  
In particular, $X$ can be represented as a sum 
\begin{equation}
\labelA{eq:split}
X 
\, = \, 
X_1 + Z,
\end{equation}
where $X_1$ has independent $N_m(0,T-I)$ columns, 
$Z$ has independent $N(0,1)$ entries, and $X_1$ and $Z$ are independent. 
It follows from the definition of $T$ that
\[
(T-I) 
\, = \, 
\mbox{diag}(\tau_1-1,\ldots,\tau_r-1,0,...,0) ,
\]
and therefore the entries in rows $r+1,\ldots,m$ of $X_1$ are equal to zero.
Thus, the sample covariance matrix 
$\widehat T_1 = n^{-1} X_1 X_1'$ of $X_1$ 
has the simple block form
\[
\vspace*{.16in}
\widehat T_1 
\, = \, 
\left[\begin{array}{c|c} \widehat T_{11} & 0 \\\hline 0 & 0\end{array}\right]
\vspace*{-.22in}
\] 
where $\widehat T_{11}$ is an $r \times r$ matrix equal to 
the sample covariance of the first $r$ rows of $X_1$.  
It is clear from the block structure that the first 
$r$ eigenvalues of $\widehat T_1$ are equal to the
eigenvalues of $\widehat T_{11}$, and that the remaining 
$(m - r)$ eigenvalues of $\widehat T_1$ are equal to zero.  
The size of $\widehat T_{11}$ is fixed, and therefore as $n$ tends to infinity,
its entries converge in probability to those of 
$\mbox{diag}(\tau_1-1,\ldots,\tau_r-1)$. In particular,
\begin{equation}
\labelA{eq:x1matconv}
\big\| \frac{1}{n} X_1 X_1' - (T-I) \big\|_F^2 \,\topr\, 0.
\end{equation}
  Consequently, 
for each $j=1,\ldots,r$,
as $n$ tends to infinity
\begin{equation}
\labelA{eq:svconv}
\lambda_j^2(n^{-1/2}X_1) 
\, = \,  
\lambda_j(\widehat T_1) 
\, = \, 
\lambda_j(\widehat T_{11}) 
\,\topr\,
\tau_j-1
\end{equation}
and 
\begin{equation}
\labelA{eq:sveconv1}
\big\langle u_j(\widehat T_{11}), e_j \big\rangle^2 \, \topr \, 1 ,
\end{equation}
where $e_j$ is the $j$-th canonical basis element in $\reals^r$.
An easy argument shows that $u_j(n^{-1/2}X_1) = u_j(\widehat T_1)$,
and it then follows from (\ref{eq:sveconv1}) that 
\begin{equation}
\labelA{eq:sveconv2}
\big\langle u_j(n^{-1/2}X_1), e_j \big\rangle^2 \, \topr \, 1 ,
\end{equation}
where $e_j$ is the $j$-th canonical basis element in $\reals^m$.

\subsubsection{Proof of Proposition \ref{prop:SPM}}

Proposition \ref{prop:SPM} is derived from existing results 
on the limiting singular values 
of $\widehat T$ in the spiked population model.  These results are
summarized in the following theorem, which is a combination of
Theorems 1.1, 1.2 and 1.3 in \cite{baik2006els}.


\begin{thmref}
\labelA{thm:BS}
If $\widehat T$ is derived from the spiked population model with parameters
$\tau_1, \ldots, \tau_r > 1$, then for $j=1,\ldots,r$, as $n \to \infty$
\[ 
\lambda_j(\widehat T) 
\, \topr \,
\left\{
\begin{array}{lll}
\tau_j + c \frac{\tau_j}{\tau_j-1} 		& \mbox{ if } &  \tau_j > 1 + \sqrt{c}\\
(1 + \sqrt{c})^2 					& \mbox{ if } & 1 < \tau_j \leq 1 + \sqrt{c}  \\
\end{array}
\right.
\]
The remaining sample eigenvalues $\lambda_{r+1}(\widehat T),\ldots,$ 
$\lambda_{m \wedge n}(\widehat T)$
are associated with the unit eigenvalues of $T$, their empirical
distribution converges weakly to the Mar{\v{c}}enko-Pastur distribution.  
\end{thmref}

We also require the following inequality of \cite{mirsky1960symmetric}.

\begin{thmref}
\labelA{thm:Mirsky}  
If $B$ and $C$ are $m \times n$ matrices then
$\sum_{j=1}^{m \wedge n} [ \lambda_j(C) - \lambda_j(B) ]^2 \, \leq \, \| C-B \|_F^2$.
\end{thmref}

\vskip.1in

\begin{proof}[Proof of Proposition \ref{prop:SPM}]

Fix $n \geq 1$ and let $Y$ follow the asymptotic reconstruction 
model (\ref{eq:asmodel}), 
where the signal matrix $A$ has fixed rank $r$ and non-zero singular values $\lambda_1(A), \ldots, \lambda_r(A)$.  
Based on orthogonal invariance of the matrix reconstruction problem, without loss of generality, we will assume that the
signal matrix $A = \mbox{diag}(\lambda_1(A), \ldots, \lambda_r(A),$ $0, \ldots, 0)$.

We begin by considering a spiked population model whose 
parameters match those of the matrix reconstruction model.
Let $X$ have the same dimensions as $Y$ and be derived 
from a spiked population model with covariance matrix $T$ 
having $r$ non-unit eigenvalues
\begin{equation}
\label{alpha}
\tau_j = \lambda_j^2(A) + 1, \ \ j = 1,\ldots, r.
\end{equation}
As noted above, we may represent $X$ as $X = X_1 + Z$, 
where $X_1$ has independent $N(0,T-I)$ columns, $Z$ has independent
$N(0.1)$ entries and $X_1$ is independent of $Z$. Recall that the limiting relations
(\ref{eq:x1matconv})-(\ref{eq:sveconv2}) hold for this representation.

The matrix reconstruction problem and spiked population model may
be coupled in a natural way.
Let random orthogonal matrices $U_1$ and $V_1$ be defined for each 
sample point in such a way that $U_1 D_1 V_1'$ is the SVD of $X_1$. 
By construction, the matrices $U_1, V_1$ depend only on $X_1$, 
and are therefore independent of $Z$.  Consequently $U_1' Z V_1$ has the 
same distribution as $Z$. 
If we define $\tilde W = U_1' Z V_1$, then $\tilde Y = A + n^{-1/2}\tilde W$ 
has the same distribution as the observed matrix $Y$ in the 
matrix reconstruction problem.

We apply Mirsky's theorem with $B = \tilde Y$ and 
$C = n^{-1/2} U_1'XV_1$ in order to bound the difference 
between the singular values of $\tilde Y$ and those
of $n^{-1/2} X$:
\begin{equation*}
\begin{split}
\sum_{j=1}^{m \wedge n} \big[ \lambda_j(n^{-1/2}X) - \lambda_j(\tilde Y) \big]^2 
\, & \leq \, 
\big\|n^{-1/2}U_1' X V_1 - \tilde Y  \big\|_F^2
\\ & = \,
\big\|(n^{-1/2}U_1' X_1 V_1 - A) + n^{-1/2}(U_1' Z V_1  - \tilde W)    \big\|_F^2
\\[0.07in] & = \,
\big\|n^{-1/2}U_1' X_1 V_1 - A \big\|_F^2
\\[0.03in] & = \,
\sum_{j=1}^{m \wedge n} \big[ \lambda_j(n^{-1/2}U_1'X_1V_1) - \lambda_j(A) \big]^2
\\ & = \,
\sum_{j=1}^{m \wedge n} \big[ \lambda_j(n^{-1/2} X_1) - \lambda_j(A) \big]^2 
.
\end{split}
\end{equation*}
The first inequality above follows from Mirsky's theorem and the fact that
the singular values of $n^{-1/2}X$ and $n^{-1/2}U_1' X V_1$ are 
the same, even though $U_1$ and $V_1$ may not be independent of $X$.  
The next two equalities follow by expanding $X$ and $\tilde Y$,
and the fact that $\tilde W = U_1' Z V_1$. 
The third equality is a consequence of the fact that both $U_1'X_1V_1$ and $A$
are diagonal, and the final equality follows from the equality of the singular values
of $X_1$ and $U_1' X_1 V_1$.
In conjunction with (\ref{eq:svconv}) and (\ref{alpha}), the last display 
implies that
\[
\sum_j \big[ \lambda_j(n^{-1/2}X) - \lambda_j(\tilde Y) \big]^2 \topr 0 .
\]
Thus the distributional and limit results for the eigenvalues of 
$\widehat T = n^{-1} XX'$ hold also for the eigenvalues of 
$\tilde Y\tilde Y'$, and therefore for $YY'$ as well. 
The relation $\lambda_j(Y) = \sqrt{\lambda_j(YY')}$ completes
the proof.
\end{proof}

\subsubsection{Proof of Proposition \ref{prop:Paul}}

Proposition \ref{prop:Paul} may be derived from existing results 
on the limiting singular vectors of the sample covariance 
$\widehat T$ in the spiked population model. These results are summarized in 
Theorem \ref{thm:Paul} below. The result was first established for Gaussian
models and aspect ratios $0<c<1$ by \cite{paul2007asymptotics}.  
\cite{nadler2008finite} extended Paul's results to $c>0$.  Recently 
\cite{lee2010convergence} further extended the theorem to $c \geq 0$
and non-Gaussian models.

\begin{thmref} 
\labelA{thm:Paul}
If $\widehat T$ is derived from the spiked population model with 
distinct parameters
$\tau_1 > \cdots > \tau_r > 1$, then for $1 \leq j \leq r$,
\[
\big\langle u_j(\widehat T),u_j(T) \big\rangle^2
\, \topr \,
\left\{
\begin{array}{lll}
\Big(1-\frac{c}{(\tau_j-1)^2}\Big) \,/\, \Big( 1 + \frac{c}{\tau_j-1}\Big)
								& \mbox{ if } &  \tau_j > 1 + \sqrt{c}\\
0 								& \mbox{ if } &  1< \tau_j \leq 1 + \sqrt{c}\\
\end{array}
\right.
\]
Moreover, for $\tau_j > 1 + \sqrt{c}$ and $k\neq j$ such that $1 \leq k \leq r$ we have
\[
	\big\langle u_j(\widehat T),u_k(T) \big\rangle^2 \topr 0.
\]
\end{thmref}
Although the last result is not explicitly stated in \cite{paul2007asymptotics}, it follows immediately from the central limit theorem for eigenvectors \citep[Theorem 5,][]{paul2007asymptotics}.

We also require the following result, which is a special case of an inequality of Wedin \citep{wedin1972perturbation, stewart1991perturbation}.

\begin{thmref}
\labelA{thm:Wedin}  
Let $B$ and $C$ be $m \times n$ matrices and let $1 \leq j \leq m \wedge n$.
If the $j$-th singular value of $C$ is
separated from the singular values of $B$ and bounded away from zero,
in the sense that for some $\delta > 0$
\[
\min_{k\neq j}\big|\lambda_j(C) - \lambda_k(B)\big| \,>\, \delta
\qquad
\mbox{ and }
\qquad
\lambda_j(C) \,>\, \delta
\]
then 
\[
\big\langle u_j(B), u_j(C) \big\rangle^2 \, + \, \big\langle v_j(B), v_j(C) \big\rangle^2 
\,\geq\,
2 - \frac{2\| B-C \|_F^2}{\delta^2} .
\]
\end{thmref}

\begin{proof}[Proof of Proposition \ref{prop:Paul}:] 

Fix $n \geq 1$ and let $Y$ follow the asymptotic reconstruction 
model (\ref{eq:asmodel}), where the signal matrix $A$ 
has fixed rank $r$ and non-zero singular values $\lambda_1(A), \ldots, \lambda_r(A)$.  
Assume without loss of generality that $A = \mbox{diag}(\lambda_1(A), \ldots, \lambda_r(A),$ $0, \ldots, 0)$.

We consider a spiked population model whose parameters match those of the matrix reconstruction problem and couple it with the matrix reconstruction model exactly as in the proof of Proposition \ref{prop:SPM}.
In particular, the quantities $\tau_j, T, X, X_1, Z, U_1, V_1, \tilde W,$ and $\tilde Y$ are as in the proof of Proposition \ref{prop:SPM} and the preceding discussion.

Fix an index $j$ such that $\lambda_j(A) > \sqrt[4]{c}$ and thus $\tau_j > 1 + \sqrt{c}$.
We apply Wedin's theorem with $B = \tilde Y$ and $C = n^{-1/2}U_1' X V_1$.
There is $\delta>0$ such that both conditions of Wedin's theorem are satisfied for the given $j$ with probability converging to 1 as $n \to \infty$. The precise choice of $\delta$ is presented at the end of this proof. It follows from Wedin's theorem and inequality $\langle v_j(B), v_j(C) \rangle^2 \leq 1$ that 
\[
\big\langle u_j(\tilde Y), u_j(n^{-1/2}U_1' X V_1) \big\rangle^2
\,=\,
\big\langle u_j(B), u_j(C) \big\rangle^2 
\,\geq\,
1-\frac{2\| B-C \|_F^2}{\delta^2}
\]
It is shown in the proof of Proposition \ref{prop:SPM} that 
$
\|B-C\|_F^2 \,=\, \|n^{-1/2}U_1' X V_1 - \tilde Y  \|_F^2 \,\topr\, 0
$ as $n \to \infty$.
Substituting $u_j(n^{-1/2}U_1' X V) = U_1' u_j(X)$ then yields
\begin{equation}
\labelA{eq:fromwedin}
	\big\langle u_j(\tilde Y) , U_1' u_j(X) \big\rangle^2 \,\topr\, 1.
\end{equation}

Fix $k=1,\ldots,r$. As $\tau_j > 1+\sqrt{c}$ \, Theorem~\ref{thm:Paul} shows that $\langle u_j(\widehat T), e_k \rangle^2$ has a non-random limit in probability,
which we will denote by $\theta_{jk}^2$. 
The relation $\tau_j = \lambda_j^2(A) + 1$ \,implies that\, $\theta_{jk}^2 = 
[1-c \lambda_j^{-4}(A)] / [ 1 + c \lambda_j^{-2}(A)]$ \, if \, $j=k$, \,and\, $\theta_{jk}^2 = 0$  otherwise.
As $u_j(\widehat T) = u_j(X)$, it follows that
\[
	\big\langle u_j(X), e_k \big\rangle^2 \, \topr \, \theta_{jk}^2.
\]
Recall that the matrix $U_1$ consists of the left singular vectors of $X_1$, i.e. $u_k(X_1) = U_1 e_k$. It is shown in (\ref{eq:sveconv2}) that $\langle  U_1 e_k, e_k \rangle^2=\langle u_k(X_1), e_k \rangle^2  \topr 1$, so we can replace $e_k$ by $U_1 e_k$ in the previous display to obtain
\[
	\big\langle u_j(X), U_1 e_k \big\rangle^2 \, \topr \, \theta_{jk}^2.
\]
It the follows from the basic properties of inner products that
\[
	\big\langle  U_1' u_j(X), e_k \big\rangle^2 \, \topr \, \theta_{jk}^2.
\]
Using the result (\ref{eq:fromwedin}) of Wedin's theorem we may replace the left term in the inner product by $u_j(\tilde Y)$, which yields
\[
	\big\langle  u_j(\tilde Y), e_k \big\rangle^2 \, \topr \, \theta_{jk}^2.
\]
As $A = \mbox{diag}(\lambda_1(A), \ldots, \lambda_r(A), 0, \ldots, 0)$ we have $e_k = u_k(A)$. 
By construction the matrix $\tilde Y$ has the same distribution as $Y$, so it follows from the last display that 
\[
	\big\langle  u_j(Y), u_k(A) \big\rangle^2 \, \topr \, \theta_{jk}^2,
\]
which is equivalent to the statement of Proposition \ref{prop:Paul} for the left singular vectors. The statement for the right singular vectors follows from consideration of the transposed reconstruction problem.
\vskip 0.2in
Now we find such $\delta>0$ that for the fixed $j$ the conditions of Wedin's theorem are satisfied with probability going to 1.
It follows from Proposition~$\ref{prop:SPM}$ that for $k=1,\ldots,r$ the $k$-th singular value of $Y$ has a non-random limit in probability
\[
\lambda_k^*
\,=\, 
\lim \lambda_k( n^{-1/2} X ) 
\,=\, 
\lim \lambda_k(\tilde Y).
\]

Let $r_0$ be the number of eigenvalues of $A$ such that $\lambda_k(A) > \sqrt[4]{c}$ (i.e.~the inequality holds only for $k = 1,\ldots,r_0$). It follows from the formula for $\lambda_k^*$ that $\lambda_k^* > 1 + \sqrt{c}$ for $k=1,\ldots,r_0$. Note also that in this case $\lambda_k^*$ is a strictly increasing function of $\lambda_k(A)$. 
All non-zero $\lambda_j(A)$ are distinct by assumption, so all $\lambda_k^*$ are distinct for $k=1,\ldots,r_0$.
Note that $\lambda^*_{r_0+1} = 1+\sqrt{c}$ is smaller that $\lambda_{r_0}^*$. 
Thus the limits of the first $r_0$ singular values of $Y$ are not only distinct, they are bounded away from all other singular values.
 Define
\[
\delta 
\,=\,
\frac{1}{3} \min_{k=1,...,r_0} (\lambda^*_k - \lambda^*_{k+1})
\,>\,
0.
\]
For any $k = 1,\ldots,r_0+1$ the following inequalities are satisfied with probability going to 1 as $n \to \infty$ 
\begin{equation}
\label{eq:pre1}
	| \lambda_k(Y) - \lambda^*_k | \,<\, \delta
\quad
\mbox{ and }
\quad
	| \lambda_k(n^{-1/2}X) -  \lambda^*_k| \,<\, \delta.
\end{equation}
In applying Wedin's theorem to $B = \tilde Y$ and $C = n^{-1/2} U_1'XV_1$
we must verify that for any $j=1,\ldots,r_0$ its two conditions are satisfied with probability going to 1.
The first condition is $\lambda_j(C) > \delta$. When inequalities (\ref{eq:pre1}) hold 
\begin{equation}
\begin{split}
\lambda_j(C) 
\, & =\, 
\lambda_j(n^{-1/2} U_1'XV_1)
\,=\, 
\lambda_j(n^{-1/2}X)
\,>\,
\lambda^*_j - \delta
\\ & >\, 
(\lambda^*_j-\lambda^*_{j+1}) - \delta
\,>\, 
3 \delta - \delta 
\,=\,
2\delta,
\end{split}
\end{equation}
so the first condition is satisfied with probability going to 1.
The second condition is $|\lambda_j(C) - \lambda_k(B)| > \delta$ for all $k \neq j$.
It is sufficient to check the condition for $k=1,\ldots,r_0+1$ as asymptotically $\lambda_j(C) > \lambda_{r_0+1}(B)$.
From the definition of $\delta$ and the triangle inequality we get
\[
3 \delta 
\,<\,
|\lambda^*_j - \lambda^*_k| 
\,\leq\,
|\lambda^*_j - \lambda_j(n^{-1/2}X)| 
\,+\,
|\lambda_j(n^{-1/2}X) - \lambda_k(\tilde Y)| 
\,+\,
|\lambda_k(\tilde Y) - \lambda^*_j|
.
\]
When inequalities (\ref{eq:pre1}) hold the first and the last terms on the right hand side sum are no larger than $\delta$, thus
\[
3 \delta 
\,<\,
\delta \,+\,
|\lambda_j(n^{-1/2}X) - \lambda_k(\tilde Y)| 
\,+\, \delta
.
\]
It follows that the second condition $|\lambda_j(C) - \lambda_k(B)| = |\lambda_j(n^{-1/2}X) - \lambda_k(\tilde Y)| > \delta$ also holds with probability going to 1.
\end{proof}

\bibliographystyle{acmtrans}
\bibliography{RMT}

\end{document}